\documentclass[a4paper,10pt]{amsart}
\usepackage{amssymb, amsmath, graphicx, amsthm, cite}
\usepackage[margin=1in]{geometry}
\usepackage{hyperref}
\usepackage[all,cmtip]{xy}

\DeclareMathOperator{\R}{\mathbb{R}}

\DeclareMathOperator{\K}{\mathbb{K}}

\DeclareMathOperator{\Sym}{\textnormal{Sym}}

\DeclareMathOperator{\id}{\textnormal{id}}

\DeclareMathOperator{\gfrak}{\mathfrak{g}}

\DeclareMathOperator{\gr}{\textnormal{gr}}

\DeclareMathOperator{\Z}{\mathbb{Z}}
\DeclareMathOperator{\Span}{\textnormal{span}}

\DeclareMathOperator{\End}{\textnormal{End}}

\DeclareMathOperator{\ad}{\textnormal{ad}}

\DeclareMathOperator{\N}{\mathbb{N}}

\DeclareMathOperator{\FF}{\mathcal{F}}
\DeclareMathOperator{\MM}{\mathcal{M}}
\DeclareMathOperator{\NN}{\mathcal{N}}
\DeclareMathOperator{\UU}{\mathcal{U}}

\newcommand{\ip}{I_0^{(2)}}
\newcommand{\iphat}{I^{(2)}}
\newcommand{\gs}{\geqslant}
\newcommand{\ls}{\leqslant}
\newcommand{\la}{\langle}
\newcommand{\ra}{\rangle}

\newcommand{\p}{\partial}

\newtheorem{lem}{Lemma}
\newtheorem{cor}[lem]{Corollary}
\newtheorem{thm}[lem]{Theorem}

\theoremstyle{definition}
\newtheorem{defi}[lem]{Definition}
\newtheorem{rem}[lem]{Remark}

\newtheorem{ass}[lem]{Assumption}

\title{The BRST Complex of Homological Poisson Reduction}
\author{Martin M\"uller-Lennert}
\address{ETH Zurich, Zurich, Switzerland}
\email{martin.mueller-lennert@math.ethz.ch}



\date{\today}

\begin{document}

\begin{abstract}
	BRST complexes are differential graded Poisson algebras. They are associated
	to a coisotropic ideal $J$ of a Poisson algebra $P$ and provide a description
	of the Poisson algebra $(P/J)^J$ as their cohomology in degree zero. Using
	the notion of stable equivalence introduced in \cite{felder-kazhdan}, we
	prove that any two BRST complexes associated to the same coisotropic ideal
	are quasi-isomorphic in the case $P = \R[V]$ where $V$ is a
  finite-dimensional symplectic vector space and the bracket on $P$ is induced
  by the symplectic structure on $V$.
  As a corollary, the cohomology of the BRST complexes is canonically
  associated to the coisotropic ideal $J$ in the symplectic case. We do not
  require any regularity assumptions on the constraints generating the ideal
  $J$.
	We finally quantize the BRST complex rigorously in the presence of infinitely
	many ghost variables and discuss uniqueness of the quantization
	procedure. 
\end{abstract}
\maketitle
\tableofcontents

\newpage

\section{Introduction}

In the quantization of gauge systems, the so-called BRST complex plays
a prominent role \cite{henneaux-teitelboim}.  In the Hamiltonian
formalism, the theory is called BFV-theory and goes back to Batalin,
Fradkin, Fradkina and Vilkovisky \cite{fradkin-fradkina,
  batalin-fradkin1983, batalin-vilkovisky1977,
  fradkin-vilkovisky1975}. 

In the Hamiltonian formulation of gauge theory, the presence of gauge freedom
yields constraints in the phase space $M$ of the system. The gauge group
still acts on the resulting constraint surface $M_0 \subset M$. The physical
observables are the functions on the quotient $\tilde{M}$ of the constraint
surface $M_0$ by this action. One wishes to quantize those observables. In the
BRST-method, one introduces variables of non-zero degree to the Poisson algebra
$P$ of functions on the original phase space. One then constructs the so-called
BRST-differential on the resulting complex and recovers the functions
on the subquotient $\tilde{M}$ as the cohomology of that complex in degree
zero. One may then attempt to quantize the system by quantizing the
BRST-complex instead of the algebra of functions on $\tilde{M}$.

The quantization procedure involves the construction of gauge
invariant observables from the cohomology of the BRST complex
\cite{brs1976, tyutin1975}. Kostant and Sternberg gave a mathematically
rigorous description of the theory \cite{kostant-sternberg} in the case where
the constraints arise from a Hamiltonian group action on phase space. They make
certain assumptions that allow the BRST-complex to be constructed as a double
complex combining a Koszul resolution of the vanishing ideal $J$ of the
constraint surface $M_0\subset M$ with the Lie algebra cohomology of the gauge
group. In more general cases, the Koszul complex does not yield a resolution
and one has to use a much bigger Tate resolution.

More recently, Felder and Kazhdan formalized the corresponding construction in
the Lagrangian formulation of the theory \cite{felder-kazhdan}. They consider
general Tate resolutions.  The aim of this note is to perform a similar
formalization in the Hamiltonian setting. We consider 
Poisson algebras $P$ as a starting point, which arise in the Hamiltonian
viewpoint as the functions on phase space. We define the notion of a
\emph{BFV-model} for a coisotropic ideal $J\subset P$. In the Hamiltonian
theory, $J$ is given as the vanishing ideal of the constraint surface $M_0
\subset M$. We use techniques from \cite{stasheff1997, felder-kazhdan} to prove existence
of BFV-models and show that they model the Poisson algebra $(P/J)^J$
cohomologically. This latter Poisson algebra is the physically
interesting one since, in the case where $P$ are the functions on phase
space and $J$ is the vanishing ideal of the constraint surface, it
corresponds to the function on the subquotient $\tilde{M}$, which are
the physical observables of the system. The statements about the
existence of what we call BFV-models and their cohomology are known
\cite{henneaux-teitelboim}. However, a rigorous treatment of the question of
uniqueness is missing. Under certain local regularity assumptions
on the constraint functions, which for instance imply that the
constraint surface $M_0$ is smooth, a construction for a uniqueness
proof for the BRST-cohomology was given in \cite{fisch1989}. Stasheff
considers the problem from the perspective of homological perturbation
theory \cite{stasheff1997} and gives further special cases under which
such uniqueness theorems hold. For instance, he considers the case
where a proper subset of the constraints satisfy a regularity
condition. Using the notion of \emph{stable equivalence} from
\cite{felder-kazhdan}, we show that, for a symplectic polynomial
algebra $P=\R[V]$ with bracket induced from the symplectic structure on a
finite-dimensional vector space $V$, any two BRST-complexes for the same
coisotropic ideal $J \subset P$ are quasi-isomorphic.  Hence, we rigorously
prove uniqueness of the BRST-cohomology for such $P$. In contrast to previous
treatments of the problem, the assumption
on $P$ does not force the constraint surface to be smooth. Moreover, we do not
assume a subset of the constraints to be regular.

Finally, we quantize the BRST-complex. Under a cohomological assumption, we
construct a quantum BRST-charge, and discuss its uniqueness. The obstruction to
quantize lies in the second degree of the classical BRST-cohomology, while the
ambiguity lies in the first degree. We do this analysis in a
rigorous fashion. To the best of our knowledge such a rigorous treatment in our
setting for general Tate resolutions is new.

In the smooth setting, Sch\"atz has dealt with the problem in
\cite{schaetz2010}. See also \cite{bordemannwaldmann2000} for the case of a
regular moment map and \cite{herbig2007} for the case of a singular moment map
of a Hamiltonian group action in the smooth setting.

I am very grateful to my supervisor Giovanni Felder for many
stimulating discussions and mathematical support. This research was
partially supported by the NCCR SwissMAP, funded by the Swiss National
Science Foundation.

\section{BFV Models}

We work over $\K = \R$, but any field of characteristic zero will do. Let $P$
be a unital, Noetherian Poisson algebra and $J \subset P$ a coisotropic ideal.
Then the Poisson structure on $P$ induces one on $(P/J)^J$. The purpose of the
BRST complex is to model this Poisson algebra cohomologically.

Let $\MM$ be a negatively graded real vector space with finite
dimensional homogeneous components $\MM^j$. Denote its component-wise
dual by $\MM^* = \bigoplus_{j>0} (\MM^{-j})^*$. Define a Poisson
bracket on $\Sym(\MM \oplus \MM^*)$ via the natural pairing between
$\MM$ and $\MM^*$. For details of the construction we refer to
chapter~\ref{sec:super-poisson} in the appendix.  

Form the tensor product $X_0 = P \otimes \Sym(\MM \oplus \MM^*)$ of
the two Poisson algebras defined above. Let $\FF^p X_0$ denote the
ideal generated by all elements in $X_0$ of degree at least $p$. Using
the filtration defined by the $\FF^pX_0$, complete the space $X_0$ to
a graded commutative algebra $X$ with homogeneous components
\begin{align*}
  X^j = \lim_{\leftarrow p} \frac{X_0^j}{\FF^p X_0 \cap X_0^j}.
\end{align*}
Extend the bracket on $X_0$ to $X$, thus turning $X$ into a graded
Poisson algebra. Again we refer to chapter~\ref{sec:super-poisson} for
details. Denote the bracket on $X$ by $\{-,-\}$.

Set $I \subset X$ to be the homogeneous ideal with homogeneous
components
\begin{align*}
  I^j = \lim_{\leftarrow p} \frac{\FF^1 X_0 \cap X_0^j}{\FF^{p+1} X_0
    \cap X_0^j} \subset X^j.
\end{align*}
Powers of ideals are denoted with exponents in parentheses, e.g. $I^{(k)}$
refers to the $k$-th power of the ideal $I$.

An element $R \in X$ of odd degree which solves $\{R,R\}=0$ defines a
differential $d_R = \{R, -\}$ on $X$ by the Jacobi identity. If $R \in
X^1$, the differential $d_R$ induces a differential on $X/I$ since it
preserves $I$.

\begin{defi}
  A \emph{BFV model} for $P$ and $J$ is a pair $(X,R)$ where
  $(X,\{-,-\})$ is a graded Poisson algebra constructed as above and $R
  \in X^1$ is such that the following conditions hold:
  \begin{enumerate}
  \item $\{R,R\}=0$.
  \item $H^j(X/I, d_R) = 0$ for $j \neq 0$.
  \item $H^0(X/I, d_R) = P/J$.
  \end{enumerate}
  The first equation is called the \emph{classical master equation}
  and the element $R$ is called a \emph{BRST charge.}
\end{defi}
The aim of this note is to prove
\begin{thm}
  \label{thm:bfvexists} Let $P$ be a Poisson algebra and $J \subset P$
  a coisotropic ideal.  BFV models exist and in the case of $P =
  \R[V]$ with with bracket induced from the symplectic structure on a finite
  dimensional vector space $V$ the
  complexes of any two BFV-models for the same ideal $J$ are quasi-isomorphic,
  whence the cohomology $H(X, d_R)$ is uniquely determined by $J$ up to
  isomorphism.
\end{thm}
The existence of BFV models is known \cite{henneaux-teitelboim,
  fisch1989}. The problem of uniqueness has been dealt with under certain
regularity assumptions \cite{henneaux-teitelboim, fisch1989}. These
assumptions imply that the constraint surface is smooth. The novel
part is the statement that any two BRST-complexes are
quasi-isomorphic, which gives uniqueness of the BRST-cohomology as a
corollary. We prove this without assuming that the constraint surface
is smooth. For completeness we also include proofs of the already
known facts in our framework.

Finally, we quantize the BRST-charge rigorously and discuss the uniqueness of
the quantization procedure.

\section{Existence}
\label{sec:existence}
\subsection{Tate Resolutions}
\label{sec:tateresolutions}

In order to construct BFV models, we first have to construct a
suitable commutative graded algebra $X$. The odd variables are
obtained via Tate resolutions.

Let $P$ be a unital, Noetherian Poisson algebra and $J \subset P$ be a
coisotropic ideal. Tate constructed resolutions of Noetherian rings by
adding certain odd variables to the ring \cite{tate1957}.
Consider a Tate
resolution $T = P \otimes \Sym(\MM)$ of $P/J$ given by a negatively
graded vector space $\MM$ with finite-dimensional homogeneous
components together with a differential $\delta$ on $T$ of degree 1. Define
the dual $\MM^*$ degree-wise. Extend $\delta$ to $X_0 := P \otimes
\Sym(\MM \oplus \MM^*) = P \otimes \Sym(\MM) \otimes \Sym(\MM^*)$ by
tensoring with the identity. Endow $X_0$ with the natural extension of
the Poisson bracket, define the filtration $\FF^p X_0$, and extend the 
bracket to the completion $X$ as described in section
\ref{sec:super-poisson} in the appendix. We will frequently refer to
statements from that chapter.
\subsubsection{The Differential $\delta$}
\label{sec:delta}
Since $\delta$ is the identity on $\Sym(\MM^*)$, it preserves the
filtration on $X_0$. Hence it extends to the completion $X$ by remark
\ref{rem:extension-maps}. Call this extension $\delta$. The extension
has degree 1 and preserves the filtration on $X$. The extension is
still an odd derivation, whose square is zero. Since $\delta$
preserves the filtration, it defines a differential on the associated
graded mapping $\gr^p X^n$ into $\gr^p X^{n+1}$.

Define $B = P \otimes \Sym(\MM^*)$. Then $X_0 = B \otimes_P T$. Since,
by definition, the extension of $\delta$ to $X$ leaves elements in
$\Sym(\MM^*)$ fixed, we have
\begin{rem}
  \label{rem:differential-delta-assoc}
  The natural isomorphism of lemma \ref{lem:associatedgraded}
  identifies the differential $\delta$ on the associated graded with
  $1 \otimes \delta$ on $B \otimes_P T$.
\end{rem}

\subsubsection{Contracting Homotopy}
\label{sec:s}

From the Tate resolution construct a contracting homotopy $s: T\to T$
of degree $-1$. Then there exists a $\K$-linear split $P/J \to P$ and
a map $\overline{\pi} : T \to T$ which is defined as the composition
$P \otimes \Sym(\MM) \to P \to P/J \to P \to P \otimes \Sym(\MM)$ such
that
\begin{align}
  \label{eq:sdeltacommute}
  \delta s + s \delta = 1 - \overline{\pi}
\end{align}
Extend $\delta$, $s$ and
$\overline{\pi}$ to $X_0$ by tensoring with the identity on
$\Sym(\MM^*)$. From the definition of $\overline{\pi}$ we find
\begin{rem}
  \label{rem:overlinepizero}
  $\overline{\pi} : X_0 \to X_0$ is zero on monomials which contain a
  factor of negative degree.
\end{rem}
The homotopy $s$ does not act on elements in $\Sym(\MM^*)$ and hence
preserves the filtration. For the same reason $\bar{\pi}$ preserves
the filtration. Both $s$ and $\bar{\pi}$ hence naturally extend to the
extension and equation \ref{eq:sdeltacommute} is valid in $X$
too. Moreover, 
\begin{rem}
  \label{lem:spreservesiphat}
  $s$ preserves $\iphat$.
\end{rem}

\subsection{Constructing The BRST Charge}
\subsubsection{First Approximation}
\label{sec:first-approximation}

\begin{defi}
  Let $Q_0$ be the differential $\delta$ on $X/I$ considered as an
  element of $X$. 
\end{defi}
Hence the cohomological conditions for $Q_0$ to be a BRST charge are
satisfied. However, $Q_0$ does not in general satisfy the classical
master equation. We are going to prove the existence part of theorem
\ref{thm:bfvexists} by adding correction terms to $Q_0$.

An explicit description of $Q_0$ is the following: Let $e_i$ be a
homogeneous basis of $\MM$, $e^*_i$ its dual basis. Set $d_i := \deg
e_i = - \deg e^*_i \equiv \deg e^*_i \pmod 2.$ Assume that $i
\leqslant j$ implies $d_i \geqslant d_j$.  Define $Q_0 := \sum_j
(-1)^{1+d_j} e^*_j \delta(e_j)$. By lemma \ref{lem:filtrationsums},
this defines an element of $X^1$. For each $p$, let $L_p$ be an
integer with $ \{ j \in \N : -d_j \leqslant p-1 \} = \{ 1, \dots,
L_p\}$ so that $ (q_0)_p := \sum_{j=1}^{L_p} (-1)^{1+d_j} e^*_j \delta(e_j) $
defines a representative of the $p$-th component of $Q_0$.  Of course,
the element $Q_0$ is independent of the choice of basis $e_j$ of
$\MM$.

\begin{lem}
  \label{lem:deltaexpression}
  We have $\delta = \sum_j (-1)^{1+d_j} \delta(e_j) \{ e^*_j, -\}$ on
  $X$ where the operator on the right hand side is well-defined.
\end{lem}
\begin{proof}
  Set $\delta' = \sum_j(-1)^{1+d_j}\delta(e_j) \{ e^*_j, -\}$. This
  defines a map on $X$: For $x \in X^n$, the elements $\{e_j^*, x\}$ are
  in $\FF^{-d_j + n} X$. Hence the sum converges by lemma
  \ref{lem:filtrationsums}. By linearity, $\delta'$ is defined on all
  of $X$. We claim that $\delta'$ is continuous on each $X^n$. Let
  $x^j = (x_{p}^j + \FF^p X_0^n)_p \in X^n$ be a sequence converging
  to zero. Fix $p$. Then there exists a $K$, independent of $j$, such
  that a $p$-th representative of $\delta'(x^{j})$ is given by
  \begin{align*}
    \sum_{k=1}^K(-1)^{1+d_k} \delta(e_k) \{e_k^*, x^j_{s_{-d_k,n}(p)}\}
  \end{align*}
  since the bracket is in $X_0^{-d_k + n}$. Now let $j_0$ be such that
  for $j\gs j_0$ and for all $k \in \{1, \dots, K\}$ we have
  $x^{j}_{s_{-d_k,n}(p)} \in \FF^{s_{-d_k, n}(p)} X_0^n$. Then the
  above representative vanishes modulo $\FF^p X_0^{n+1}$ by corollary
	\ref{cor:tcor}. Hence $\delta'$
  is continuous on $X^n$. The map $\delta'$ descends to a map on
  $X_0$ since the sum is then effectively finite since $\{e^*_j, x\}$
  becomes zero for $j$ large enough, depending on $x \in X_0$. This
  restriction agrees with $\delta$, which can be checked on generators since
  both maps are derivations. Hence $\delta' = \delta$ on each $X^n$ by
  continuity. Hence $\delta = \delta'$.
\end{proof}
\begin{lem}
  \label{lem:l0increase} For $L_0 := \{Q_0, - \} - \delta$ we have
  $L_0(\FF^p X) \subset \FF^{p+1} X$.
\end{lem}
\begin{proof}
  Fix $x \in \FF^p X^n$. Then, by lemma \ref{lem:bracketcontinuous},
  \begin{align*}
    \{Q_0, x\} = \lim_{m \to \infty} \bigg( \sum_{j=1}^m (-1)^{1+d_j} \delta(e_j)
    \{e_j^*, x\} + \sum_{j=1}^m (-1)^{1+d_j} e_j^* \{\delta(e_j), x\} \bigg).
  \end{align*}
  The first part converges to $\delta(x)$ by lemma
  \ref{lem:deltaexpression}. The second part converges by lemma
  \ref{lem:filtrationsums} and hence equals $L_0$. Fix $j$. By lemma
  \ref{lem:filtrationclosed} it suffices to prove that $e_j^*
  \{\delta(e_j), x\} \in \FF^{p+1} X$. By the derivation property it
  suffices to consider $x=e^*_l$ for some $l$. The term $\delta(e_j)$ is a
  sum of monomials whose factors have degrees in $\{d_j + 1, \dots,
  0\}$. Hence all elementary factors $e_r$ in $\delta(e_j)$ that could
  possibly kill $e^*_l$ have degree $d_l$ and get compensated by a
  factor $e^*_j$ with $\deg(e^*_j) > \deg(e_l^*)$.
\end{proof} 
Moreover, we have
\begin{lem}
  \label{lem:ind1}
  $\{ Q_0, Q_0 \} \in X^2 \cap \iphat \subset \FF^2 X \cap
  \iphat$. 
\end{lem}
\begin{proof}
  We compute $\deg \{Q_0, Q_0\} = 2 \deg Q_0 = 2$ and hence $\{Q_0, Q_0\} \in
  \FF^2 X$. For the last statement we need to calculate.
  By lemma \ref{lem:bracketcontinuous},
  \begin{align*}
    \{Q_0, Q_0\} &= \lim_{m\to\infty}(-1)^{d_j+d_k} \sum_{j,k= 1}^m \{\delta(e_j)e^*_j,
    \delta(e_k) e^*_k\}\\
    &=\lim_{m\to\infty} \sum_{j,k= 1}^m 
    \bigg( 
    2(-1)^{1+ d_k} \Big( (-1)^{1+d_j}
    \delta(e_j) \{e^*_j,
    \delta(e_k)\}\Big) e^*_k
    +(-1)^{d_j+d_k}e^*_j
    \{\delta(e_j),
    \delta(e_k)\} e^*_k
    \bigg)
  \end{align*}
  By lemma \ref{lem:deltaexpression} the first term is a sum in $k$
  with summands that contain factors $\delta(\delta(e_k))=0$ and
  hence the first term vanishes. By lemma \ref{lem:iphattopclosed},
  $\{Q_0,Q_0\} = \sum_{j,k}(-1)^{d_j + d_k} e_j^* \{\delta(e_j), \delta(e_k)\} e_k^* \in
  \iphat$. 
\end{proof}
\begin{cor}
  \label{cor:ind2}
  $\delta \{Q_0, Q_0\} \in X^3 \subset \FF^3 X$.
\end{cor}

\subsubsection{Recursive Construction}

We now inductively construct out of $Q_0$ a sequence of elements $R_n
\in X^1$ by setting
\begin{align*}
  R_n &= \sum_{j=0}^n Q_j, & Q_0& \text{ as defined above}, &
  Q_{n+1} &= -\frac{1}{2} s\{R_n, R_n\}.
\end{align*}
The elements $R_n$ have degree 1 since $Q_0$ has and $s$ is of degree
$-1$. The idea for the construction is taken
from~\cite{stasheff1997}. Also the proof of the following theorem is
adapted from that paper.
\begin{thm}
  For all $n$, $\{R_n,R_n\} \in \FF^{n+2} X \cap \iphat$, and
  $\delta \{R_n, R_n\} \in \FF^{n+3} X$.
\end{thm}
\begin{proof}
  The base step was done in lemma \ref{lem:ind1} and corollary
  \ref{cor:ind2}. We assume the statement is true for $0 \ls j \ls n$ and
  consider 
  \begin{align*}
    \{R_{n+1}, R_{n+1}\} = \{R_n, R_n\} + 2 \{R_n, Q_{n+1}\} +
    \{Q_{n+1},Q_{n+1}\}
  \end{align*}
  By construction and assumption $Q_{n+1} = -\frac{1}{2}s\{R_n, R_n\}
  \in \FF^{n+2} X^1$.  Hence, by corollary \ref{cor:increaseF},
  \begin{align*}
    \{R_{n+1},R_{n+1}\} \equiv \{R_n, R_n\} + 2 \{R_n, Q_{n+1}\} \pmod
    {\FF^{n+3} X}. 
  \end{align*}
  Expand $\{R_n, Q_{n+1}\} = \sum_{j=1}^n \{Q_j,Q_{n+1}\} + \{Q_0,
  Q_{n+1}\}$. We have, for $j \in \{1, \dots, n+1\}$ by inductive
  hypothesis, that $Q_j = -\frac{1}{2}s\{R_{j-1},R_{j-1}\} \in \FF^{j+1}
  X^1 \cap \iphat$. Hence, by lemma \ref{lem:ipbracket}, 
  \begin{align*}
    \{R_{n+1},R_{n+1}\} \equiv \{R_n, R_n\} + 2 \{Q_0, Q_{n+1}\} \pmod
    {\FF^{n+3} X}.
  \end{align*}
  We split
  $\{Q_0, Q_{n+1}\} = \delta Q_{n+1} + L_0 Q_{n+1}$ and, by lemma
  \ref{lem:l0increase},
  \begin{align*}
    \{R_{n+1},R_{n+1}\} \equiv \{R_n, R_n\} + 2 \delta Q_{n+1} \pmod
    {\FF^{n+3} X}.
  \end{align*}
  Commuting $\delta$ and $s$,
  \begin{align*}
    2 \delta Q_{n+1} = - \delta s \{R_n, R_n\} =
    s \delta \{R_n, R_n\} - \{R_n, R_n\} + \overline{\pi} \{R_n, R_n\}.
  \end{align*}
  Since $\{R_n, R_n\} \in \FF^{n+2} X^2$, we have that $\overline{\pi}
  \{R_n, R_n\} = 0$ for $n > 0$ by remark \ref{rem:overlinepizero}. For
  $n = 0$ we obtain $\overline{\pi}\{R_0,R_0\} = 0$ from $\{Q_0,Q_0\} =
  \sum_{j,k} \pm e_j^* \{\delta(e_j), \delta(e_k)\} e_k^*$ and the fact,
  that $\overline{\pi}$ is zero on $\{J,J\} \subset J$. Hence
  \begin{align*}
    \{R_{n+1},R_{n+1}\} \equiv s \delta \{R_n, R_n\} \pmod {\FF^{n+3}
      X}. 
  \end{align*}
  which vanishes modulo $\FF^{n+3}X$ by the assumption on $\delta
  \{R_n, R_n\}$.
  
  Next, by the graded Jacobi identity we have $0 =
  \{R_{n+1},\{R_{n+1},R_{n+1}\}\}$.  From lemma \ref{lem:l0increase} and
  lemma \ref{lem:ipbracket} we find that $L_{n+1} := \{R_{n+1}, - \} -
  \delta = L_0 + \sum_{j=1}^{n+1} \{Q_j, -\}$ increases filtration
  degree. Hence $L_{n+1}\{R_{n+1}, R_{n+1}\} \in \FF^{n+4} X$ and thus $
  \delta\{R_{n+1},R_{n+1}\} \in \FF^{n+4} X.$
 
  Finally, we prove $\{R_{n+1},R_{n+1}\} = \{R_n, R_n\} + 2 \{R_n, Q_{n+1}\}
  + \{Q_{n+1},Q_{n+1}\} \in \iphat$. By hypothesis $\{R_n, R_n\} \in
  \iphat$. Next $\{Q_{n+1}, Q_{n+1}\} \in \{\iphat, \iphat\} \subset
  \iphat$ by lemma \ref{lem:ipclosed}. Now, by the same lemma, for $j
  \in \{1, \dots, n\}$, $\{Q_j, Q_{n+1}\} \in \{\iphat,\iphat\}\subset
  \iphat$ and $\{Q_0, Q_{n+1}\} \in \{ I, \iphat\} \subset \iphat$
  which concludes the proof.
\end{proof}
From $Q_{n+1} = -\frac12 s\{R_n, R_n\} \in \FF^{n+2}X^1$ it follows that
the $R_n = \sum_{j=0}^n Q_j$ converge to an element $R \in X^1$ by lemma
\ref{lem:filtrationsums}. From lemma \ref{lem:bracketcontinuous}, we
obtain $\{R_n, R_n\} \longrightarrow \{R,R\}$ as $n \longrightarrow
\infty$. We obtain
\begin{cor}
  $\{R,R\} = 0$.
\end{cor}
\begin{proof}
  We have $\{R_{n+l}, R_{n+l}\} \in \FF^{n+2} X^2$ for all $l
  \geqslant 0$. Hence $\{R,R\} \in \FF^{n+2} X^2$ for all
  $n$ by lemma \ref{lem:filtrationclosed}. Hence $\{R,R\} = 0$.
\end{proof}
We also remark that $R$ as defined above satisfies $R \equiv Q_0 \pmod
\iphat$ since for $j > 0$, we have $Q_j \in \iphat$. We are left to
consider the cohomology of $d_R = \{R,-\}$ on $X/I$.
\begin{lem}
  The action of $d_R$ preserves the filtration and hence defines a
  differential on $\gr X$, which is identified with $1 \otimes
  \delta$ under the natural isomorphism of lemma
  \ref{lem:associatedgraded}. 
\end{lem}
\begin{proof}
  $R \in X^1$ and lemma \ref{lem:optimal} imply that $d_R$ preserves
  the filtration and hence descends to the associated graded. We have
  $\{Q_0,-\} = L_0 + \delta$. Since $L_0$ increases filtration degree by
  lemma \ref{lem:l0increase}, we have that $\{ Q_0, - \}$ and $\delta$ induce
  the same maps on $\gr X$.  Moreover, $K := R - Q_0 \in \iphat \cap
  X^1$ by the remark above. Hence by lemma ~\ref{lem:ipbracket}, $\{K, -\}$
  increases filtration degree and thus $d_R = \{R, -
  \}$ and $\{Q_0, -\}$ induce the same maps on the associated graded.
\end{proof}
\begin{cor}
  $H^j( X / I, d_R ) \cong P/J$ if $j = 0$ and zero otherwise.
\end{cor}
\begin{proof}
  $H^j(X/I, d_R) = H^j(\gr^0 X, d_R) \cong H^j( B^0 \otimes_P T, 1
  \otimes \delta) \cong H^j(T, \delta)$.
\end{proof}

Given a unital, Noetherian Poisson algebra $P$ with a coisotropic
ideal $J$, we thus have constructed a BFV model for $(P,J)$.

\section{Properties}
\label{sec:properties}

In this section we describe general properties of BFV models. We
postpone the discussion of their cohomology to section
\ref{sec:cohomology}. 

Let $(X,R)$ be a BFV-model of $(P,J)$ with $X$ being the completion of
$P \otimes \Sym(\MM \oplus \MM^*)$. Since $R$ is of degree one, the
differential $d_R$ preserves the filtration and hence descends to $\gr
X$. Let $\pi : X \to X/I = T = P \otimes \Sym(\MM)$ be the canonical
projection. Let $j : T \to X_0 \to X$ be the inclusion given by $t \mapsto 1
\otimes t \in \Sym(\MM^*) \otimes T = X_0$. Define $\delta = \pi
\circ d_R \circ j : T \to T$.
\begin{lem}
  The map $\delta: T \to T$ is a derivation and a differential
  of degree 1.
\end{lem}
\begin{proof}
  The derivation property follows immediately. Let $a \in T$. We have
  $(j \circ \pi - \id_X) (d_R(j(a)))\in I$ and hence $d_R( j ( \pi (
  d_R(j(a)))) = d_R( (j \circ \pi - \id_X ) (d_R(j(a)))) \in I$ is in the
  kernel of $\pi$. The statement about the degree is obvious.
\end{proof}
\begin{lem}
\label{lem:induceddiff}
  Under the identification of lemma \ref{lem:associatedgraded}, the
  differential $d_R$ induced on $gr X$ corresponds to the differential
  $1 \otimes \delta$ on $B \otimes_P T$.
\end{lem}
\begin{proof}
  Let $x \in \gr^p X$ and pick a representative $a \otimes b \in B^p
  \otimes_P T$. (It suffices to consider the case where this is a pure tensor.)
  Then $d_R(ab) = d_R(a) b + (-1)^p a d_R(b)$. The first summand is in
  $\FF^{p+1}X$ and the second is equivalent to $1 \otimes
  \delta(a\otimes b)$ modulo $\FF^{p+1}X$.
\end{proof}
Let $Q_0$ be the differential $\delta$ on $X/I$ as an element of
$X$.
\begin{rem}
  The complex $(X/I, d_R) = (T, \delta)$ is a Tate resolution of
  $P/J$. Hence the results from appendix~\ref{sec:super-poisson} and
  sections~\ref{sec:tateresolutions} and~\ref{sec:first-approximation} apply.
\end{rem}
\begin{lem}
  \label{lem:q0}
  We have $R \equiv Q_0 \pmod \iphat$. Moreover, $\{R, -\} \equiv \{Q_0,
  -\} \pmod I$.
\end{lem}
\begin{proof}
  We have $\{R, - \} \equiv \{Q_0, -\} \pmod I$ by construction. Expand
  $R-Q_0 = \sum_{j \geqslant 0} h_j$ with $h_j \in B^j \otimes_P
  T^{1-j}$. Such a decomposition exists by lemma
  \ref{lem:expansion}. Decompose $h_j = \alpha_j + \beta_j$ with
  $\alpha_j \in B^j \otimes_P T^{1-j} \cap \ip$ and $\beta_j \in B^j
  \otimes_P T^{1-j} \setminus \ip$. Let $\{e_k^{(l)}\}_k$ be a basis
  of $\MM^{-l}$ with dual basis $\{ {e_k^{(l)}}^*\}_k$. By the
  Leibnitz rule, $\sum_j \{\beta_j, e_k^{(l)}\} = \{R-Q_0, e_k^{(l)}\} -
  \sum_j \{\alpha_j, e_k^{(l)}\} \in I$. Expand each $\beta_j = \sum_s
  a_{j,s} {e_s^{(j)}}^*$ with $a_{j,s} \in T$. We obtain $a_{l,k} =
  (-1)^{1+l}\sum_j\{\beta_j, e_k^{(l)}\} \in I$, hence all $a_{l,k}$ vanish.
\end{proof}

\section{Uniqueness}
\label{sec:uniqueness}

Fix a unital, Noetherian Poisson algebra $P$ and a coisotropic ideal
$J$.  In a first step, we prove that two BFV models for $(P,J)$
related to the same Tate-resolutions have isomorphic
cohomologies. This is a known fact \cite{henneaux-teitelboim,
  fisch1989} and is presented in sections~\ref{sec:gauge}
to~\ref{sec:gaugeuniqueness}. The key tool will be the notion of gauge
equivalences. In a second step, we prove that BFV models for $(P = \R[V],J)$,
$V$ a finite-dimensional symplectic vector space, on different spaces $X$ have
isomorphic cohomologies
too. We present this result in sections~\ref{sec:trivial}
to~\ref{sec:stableuniqueness}. Here, the key tool will be the notion
of stable equivalence, introduced in the corresponding Lagrangian
setting in \cite{felder-kazhdan}. The novel part is that we do not
require regularity assumptions, that would imply that the constraint surface is
smooth.

\subsection{Gauge Equivalences}
\label{sec:gauge}

We adapt the language of \cite{felder-kazhdan} and call the elements in
$\gfrak = X^0 \cap \iphat$ generators of gauge equivalences. Different
BRST-charges for the same Tate resolution will be related by these
equivalences.
\begin{lem}
  \label{lem:exponentiates}
  The set of generators of gauge equivalences $\mathfrak{g}$ is a closed subset
  which forms a Lie algebra acting nilpotently on $X/\FF^pX$ via the adjoint
  representation. The Lie algebra
  $\ad(\gfrak)$ exponentiates to a group $G$ acting on $X$ by
  Poisson automorphisms.
\end{lem}
\begin{proof}
  By lemma \ref{lem:iphattopclosed}, the set is closed. By lemma
  \ref{lem:ipclosed} and the fact that $\{X^0, X^0\} \subset X^0$, this
  is a Lie algebra. By corollary \ref{cor:x0preservesfp}, $\gfrak$ acts on
  $X/\FF^p X$. By lemma \ref{lem:ipnil}, this action is nilpotent. Hence $\ad
	\gfrak$ exponentiates to a group acting on $X$ by vector space automorphisms.
	Since $\ad \gfrak$ consists of derivations both for the product and the
	bracket, those automorphisms are Poisson.

\end{proof}
The elements of $G$ are called gauge equivalences.
\begin{lem}
\label{lem:gaugestill}
  For $x \in X^1$ and a gauge equivalence $g$ we have $gx \equiv x
  \pmod { \iphat }$.
\end{lem}
\begin{proof}
  Let $c \in X^0 \cap \iphat$ be a generator. Then $ gx - x =
  \sum_{j>0} \frac{1}{j!} \ad_c^j x \in \iphat $ by lemmas
   \ref{lem:iphattopclosed} and \ref{lem:iqx1}.
\end{proof}

\subsection{Uniqueness for Fixed Tate Resolution}
\label{sec:gaugeuniqueness}
In this section we prove that given two solutions $R,R'$ of the
classical master equation in the same space $X$ which induce the same
map on $X/I$ are related by a gauge equivalence. Since by lemma
\ref{lem:exponentiates}, gauge equivalences are Poisson automorphisms,
this implies that they have isomorphic cohomologies. We use known
techniques, which are adapted from \cite{felder-kazhdan}.

\begin{rem}
  \label{rem:gaugecme}
  If $R$ solves $\{R,R\} = 0$ and $g \in G$ is a gauge equivalence,
  then also $\{gR,gR\} = 0$.
\end{rem}
We now discriminate elements in the associated graded $\gr^p X^n$ according to how many
positive factors they contain at least by defining $A^n_{p,q}:= \{ v \in \gr^p
X^n : \text{$v$ has representative in $I^{(q)}$} \}$. From the proof
of lemma \ref{lem:associatedgraded} we see that $A^n_{p,q}$ can be
identified with $(B^p \cap I_0^{(q)}) \otimes_P T$, where $I_0 = \FF^1 X_0$. We
now use remark \ref{rem:differential-delta-assoc} to see that $A^\bullet_{p,q}$
is a subcomplex and bound its cohomology:
\begin{lem}
  \label{lem:assoccoho}
  Fix $p$ and $q$.   We have $H^j( A^\bullet_{p,q}, \delta ) = 0$ for $j < p$.
\end{lem}
\begin{proof}
  From remark \ref{rem:differential-delta-assoc}, we have
  $H^j( A^\bullet_{p,q}, \delta)  \cong H^j( (B^p \cap
   I_0^{(q)}) \otimes_P T, 1\otimes \delta)$. Now we may factor
  this space into $(B^p \cap I_0^{(q)}) \otimes_P H^{j-p}(T,
  \delta)$ since $B^p \cap I_0^{(q)}$ is a free $P$-module. For $j<p$, the
  second factor vanishes, since $T$ is a resolution of $P/J$.
\end{proof}
\begin{lem}
  \label{lem:gaugeinduction}
  Fix $p \geqslant 2$. Let $R, R' \in X^1$ be two solutions of the
  classical master equation which induce the same maps on $X/I$. Then,
  for $2 \leqslant q \leqslant p$, we have that $R \equiv R' \pmod{
    I^{(q)} \cap \FF^p X^1 + \FF^{p+1} X^1}$ implies the existence of
  a gauge equivalence $g$ with generator $c \in \FF^p X^0 \cap \iphat$
  such that $gR \equiv R' \pmod{I^{(q+1)} \cap \FF^p X^1 + \FF^{p+1}
    X^1}$. Moreover, the element $gR \in X^1$ sill satisfies the classical
	 master equation and induces the same map on $X/I$ as $R$ and $R'$.
\end{lem}
\begin{proof}
  Let $\delta$ be the common differential on $X/I$ and $Q_0$ be the
  map $\delta$ as an element of $X$. Hence $R \equiv Q_0 \equiv R'
  \pmod \iphat$ by lemma \ref{lem:q0}.  Define $v := R-R' \in I^{(q)}
  \cap \FF^p X^1+ \FF^{p+1} X^1 \subset \FF^p X^1$. We have
  \begin{align*}
    0 = \{R + R', R- R'\} = 2 \{Q_0, v\} + \{R-Q_0, v\} +
    \{R'-Q_0, v\} \equiv 2 \{Q_0, v\} \pmod {\FF^{p+1} X^2}
  \end{align*}
  by lemma \ref{lem:ipbracket}. By lemma \ref{lem:induceddiff}, the
  maps $d_R$ and $d_{R'}$ also induce the same map on of $\gr
  X$ which we denote by $\delta$ too. Since $\delta v = \{Q_0, v\} - L_0 v
  \equiv \{Q_0, v\} \pmod{\FF^{p+1} X^2}$ by lemma \ref{lem:l0increase}, the
  above implies $\delta v \equiv 0 \pmod{ \FF^{p+1} X^2}$. Hence $v$ defines a
  cocycle $\bar{v}$ in $\gr^p X^1$. We have $p>1$. By lemma
  \ref{lem:assoccoho} there exists $\bar{c} \in \gr^p X^{0}$ with $\delta
  \bar{c} = \bar{v}$ and a corresponding representative $c \in
  \FF^p X^0 \cap I^{(q)}$, so that $\delta c \equiv v \pmod{ \FF^{p+1} X^1}$.
  This $c$ will be the generator of the gauge equivalence we seek. Set $g :=
  \exp \ad_c$. We have 
  \begin{align*}
    g R - R' = v + \sum_{j=1}^\infty \frac{1}{j!} \ad_c^j R 
    \equiv 
	 \delta(c) - d_R(c) + \sum_{j=2}^\infty \frac{1}{j!} \ad_c^j R \equiv
	 \sum_{j=2}^\infty \frac{1}{j!} \ad_c^j R  \pmod{ \FF^{p+1}
      X^1} .
  \end{align*}
  From lemma \ref{lem:optimal}, we know that this sum is in $\FF^p
  X^1$. We are left to show that the sum is in $I^{(q+1)}$.
  By lemma \ref{lem:iqx1} we have $\ad_c R \in I^{(q)}$. By lemma
  \ref{lem:ipnil} we obtain $\ad_c^j R \in I^{(q+1)}$ for all $j
  \geqslant 2$ since $q \geqslant 2$.

  By remark \ref{rem:gaugecme}, $g R$ still
  satisfies the classical master equation and $gR \equiv Q_0 \pmod{ \iphat}$
  by lemma \ref{lem:gaugestill}, whence all maps $R, R', gR$ induce the same
  map on $X/I$.
\end{proof}
\begin{thm}
  \label{thm:gaugeequivalence}
  Let $R, R' \in X^1$ be solutions of the classical master equation
  with differentials inducing the same maps on $X/I$. Then there
  exists a gauge equivalence $g \in G$ with $R' = g R$.
\end{thm}
\begin{proof}
  First we inductively construct a sequence of gauge equivalences
  $g_2, g_3, \dots$ such that for all $p \geqslant 2$ we have $g_p \cdots g_2 R
  \equiv R' \pmod {\FF^{p+1} X^1 \cap \iphat}$. By lemma \ref{lem:gaugestill},
  it suffices to ensure that $g_p \cdots g_2 R \equiv R' \pmod {\FF^{p+1}
    X^1}$
  
	 For $p = 2$  note that $R - R' \in I^{(2)} \subset (I^{(2)} \cap \FF^2 X^1)
	 + \FF^3 X^1$ by lemma \ref{lem:q0}. Now apply lemma
	 \ref{lem:gaugeinduction} with $q = p$ to obtain $g_2$. 
	 
	 Next, assume the $g_2, \dots, g_p$ have been constructed to fulfill
  \begin{align*}
    R'' := g_p \cdots g_2 R \equiv R' \pmod{ \FF^{p+1} X^1 \cap \iphat}.
  \end{align*}
  By remark \ref{rem:gaugecme}, $R''$ solves the classical master
  equation. Moreover $R'' \equiv Q_0 \pmod{ \iphat}$ by lemmas \ref{lem:q0} and
  \ref{lem:gaugestill}. Hence the pair $(R'', R')$
  satisfies the requirements of lemma \ref{lem:gaugeinduction} with
  $q=2$. We obtain a gauge equivalence $g_{p+1,2}$ with generator
  $c_{p+1, 2} \in \FF^{p+1} X^0 \cap I^{(2)}$ and
  \begin{align*}
    g_{p+1,2} R'' \equiv R' \pmod { I^{(3)} \cap \FF^{p+1} X^1 + \FF^{p+2}
      X^1} .
  \end{align*}
  If we continue to apply the lemma for $q = 3, \dots, p+1$ we obtain gauge
  equivalences $g_{p+1,3}, \dots, g_{p+1,p+1}$ with generators
  $c_{p+1, 3}, \dots, c_{p+1,p+1} \in \FF^{p+1} X^0 \cap I^{(2)}$ such that
  \begin{align*}
    g_{p+1, p+1} \cdots g_{p+1, 2} R'' \equiv R' \pmod{\FF^{p+2} X^1
      \cap \iphat}
  \end{align*}
  Set $g_{p+1} := g_{p+1,p+1} \cdots g_{p+1,2}$. The construction of the
  sequence is complete.
  
  We claim that $\lim_{m \to \infty} g_m g_{m-1} \cdots g_2$ converges
  point-wise to a gauge equivalence $g$. Since all generators $c_{m,
  j}$ are in $\FF^{m} X^0 \cap I^{(2)}$ and this set is closed under the
  bracket, the Campbell-Baker-Hausdorff formula implies that the generator
  $c_m$ of $g_m$ is also in $\FF^{m} X^0 \cap I^{(2)}$. Now denote the
  generator of $g_m \cdots g_2$ by $\gamma_m$. We have $\gamma_m \in I^{(2)}$
  by the CBH formula. Moreover, the CBH
  formula implies that the generator $\gamma_{m+1}$ of $g_{m+1} g_m
  \cdots g_2$ satisfies
  \begin{align*}
    \gamma_{m+1} = c_{m+1} + \gamma_{m} + \text{higher terms}
  \end{align*}
  where ``higher terms'' are terms involving commutators of $c_{m+1}$
  and $\gamma_m$ where each contains at least one instance of
  $c_{m+1} \in \FF^{m+1} X^0$. Since $\gamma_m \in X^0$ all these terms
  are in $\FF^{m+1} X^0$. Hence
  \begin{align*}
    \gamma_{m+1} \equiv \gamma_m \pmod {\FF^{m+1} X^0}
  \end{align*}
  Hence there exists $\gamma \in X^0$ with $\gamma_m \to \gamma$ as $m
  \to \infty$. We set $g := \exp \ad_\gamma$. By lemma
  \ref{lem:iphattopclosed} this element defines a gauge
  equivalence. We claim that $\exp \ad_{\gamma_m} = g_m \dots g_2 \to g$
  point-wise. Let $x \in X^n$. Then
  \begin{align*}
    \exp \ad_{\gamma_m} x - \exp \ad_{\gamma} x & = \{\gamma_m -
    \gamma, x\} + \frac12 \{\gamma_m,\{\gamma_m,x\}\} - \frac12 \{\gamma,
    \{\gamma,x\}\} + \cdots
  \end{align*}
  Modulo a fixed $\FF^k X$, this sum is finite and the number of terms
  does not depend on $m$ since all $\gamma_m$ are at least in
  $\iphat$. Since $\gamma_m \to \gamma$ and the bracket is continuous
  in fixed degree by lemma \ref{lem:bracketcontinuous}, we obtain the
  claim.

  Finally $\exp \ad_{\gamma_{m+l}} R - R' \in \FF^{m}
  X^1$ implies $g R - R' \in \FF^{m} X^1$ for all $m$ which implies
  $gR = R'$.
\end{proof}

\subsection{Trivial BFV Models}
\label{sec:trivial}

The key construction in the proof of uniqueness for different spaces
$X$ in theorem \ref{thm:bfvexists} is the notion of stable
equivalence. The idea of adding variables that do not change the
cohomology was already present in \cite{henneaux-teitelboim}. It was
first explicitly formalized in \cite{felder-kazhdan} in a similar situation
in the Lagrangian setting. Roughly speaking, one proves that different
BRST complexes for the same pair $(P,J)$ are quasi-isomorphic by adding
more variables of non-zero degree. This is formalized by taking
products with so-called trivial BFV models.

Let $P = \R$ with zero bracket and $J = 0$. Then $P$ is a unital,
Noetherian Poisson algebra and $J$ is a coisotropic ideal. Let $\NN$
be a negatively graded vector space and $\NN[1]$ the same space with
degree shifted by $-1$. Define the differential $\delta$ on $\MM = \NN
\oplus \NN[1]$ by $\delta( a \oplus b) = b \oplus 0$. Set $T = P
\otimes \Sym(\MM)$ and extend $\delta$ to an odd, $P$-linear derivation on
$T$.
\begin{lem}
  \label{lem:trivial-cohomology}
  The complex $(T, \delta)$ has trivial cohomology 
  and hence defines a Tate resolution of $P/J = \R$.
\end{lem}
\begin{proof}
  On $\MM$ there is a map $s( a \oplus b ) = 0 \oplus a$ with $s
  \delta + \delta s = \id_{\MM}$. Extend $s$ to an odd, $P$-linear derivation
  on $T$. Then $s \delta + \delta s$ is an even derivation on $T$ which is the
  identity on $\MM$ and hence
  \begin{align*}
    s \delta + \delta s = k \id \quad \text{on $P \otimes \Sym^k(\MM)$}.
  \end{align*}
  Since both $s$ and $\delta$ preserve the $k$-degree, we have
  \begin{align*}
    H^j( T, \delta) = \oplus_k H^j( P \otimes \Sym^k(\MM), \delta) =
    H^j( P \otimes \Sym^0(\MM), \delta) =
    \begin{cases}
      P, & \text{if $j=0$}\\
      0, & \text{otherwise}
    \end{cases}
  \end{align*}
\end{proof}
Complete the space
$Y_0 = P \otimes \Sym(\MM \oplus \MM^*)$ to the space $Y$. Let $e_j$
be a homogeneous basis of $\MM$ such that $\delta(e_j) = e_k$ for some
$k$ depending on $j$. Define $Q_0 = \sum_j (-1)^{1+d_j} e^*_j \delta(e_j)$ as in
section \ref{lem:l0increase}. Since $\{Q_0, Q_0\} = \sum_{j,k} \pm
e^*_j\{\delta(e_j), \delta(e_k)\}e^*_k = 0$, the construction of section
\ref{sec:existence} yields the BRST charge $S = Q_0$. Hence $(Y,S)$ is
a BFV model for $(P,J) = (\R,0)$. BFV models arising from this
construction are called \emph{trivial}. 
\begin{lem}
  \label{lem:trivial-induced}
  For trivial BFV models, $d_S$ equals the induced map of $\Delta = \delta
  \oplus \delta^* : \MM \oplus \MM^* \to \MM \oplus \MM^*$ on $Y$, where the
  dual differential $\delta^* : \MM^* \to \MM^*$ is given by $\delta^* ( u ) =
  (-1)^{\deg u} u \circ \delta$, i.e. $\delta^*( a \oplus b) = (-1)^{j} 0
  \oplus a$ on $(\MM^*)^j$. Conversely, the map $d_S$ induces a differential on
	$Y_0$ which coincides with the induced map of $\delta \oplus \delta^*$ on
	$Y_0$. 
\end{lem}
\begin{proof}
  By acting on the generators $e_j$ and $e_j^*$ defined above one sees
  that the induced differential $d_S$ on $Y_0$ equals
	$\Delta = \delta \oplus \delta^*$. Since $Y_0$ is dense in $Y$ and both maps
	are continuous the first claim follows. The second claim follows from the
	observation that the sum $d_S(x) = \sum_j \pm \{e_j \delta(e_j), x\}$ is
	effectively finite if $x \in Y_0$.
\end{proof}
\begin{lem}
\label{lem:trivial-bfv-coho}
  We have $H^j( Y, d_S ) = 0$ for $j \neq 0$ and $H^0( Y, d_S ) = \R$.
	The same statement holds if we replace $Y$ by $Y_0$.
\end{lem}
\begin{proof}
  Since the cohomology of the complex $(\MM \oplus \MM^*, 
	\Delta)$ is trivial, there is a map $s : \MM \oplus \MM^* \to \MM \oplus
	\MM^*$ of degree $-1$ with $\Delta s + s \Delta = \id$. Its extension to
	$Y_0$ as a derivation thus
	satisfies 
	\begin{align}
					\label{eq:ch}
    s \Delta + \Delta s = j \id \quad \text{on $Y_0^{n,j}$}. 
  \end{align}
	Since both maps $\Delta$ and $s$ preserve form degree, we obtain the
	statement about $H^j(Y_0, d_S)$. 

  Let $x \in Y^n$ with $d_S x = 0$. By lemma
  \ref{lem:form-degree-decomposition}, there are $x_j \in Y^n$ of form
  degree $j$ with $\sum_j x_j = x$. By continuity of the bracket, $ 0
  = \sum_{j} d_S x_{j}$.  By lemma \ref{lem:form-degree-sum}, $d_S
  x_{j} = 0$ for all $j$, since $d_S$ preserves form degree. 
	Lift $\Delta$ and $s$ to $Y$. Then equation \ref{eq:ch} is still valid on
	$Y^{n,j}$. 
  For $j >0$, there are $y_{j} \in Y^{n-1, j}$ with $d_S
  y_{j} = x_{j}$. By lemma \ref{lem:form-degree-convergence}, the
  element $y = \sum_{j> 0} y_{j}$ is well defined and
  \begin{align*}
    x = \sum_{j>0}x_{j} + x_{0} = d_S y + x_{0}
  \end{align*} 
  with $x_{0} \in Y^{n,0}$. For $n \neq 0$ this is the empty set
  and hence $x$ is exact. For $n = 0$ this set is $\R$. We are left to
  show that two distinct $d_S$-closed elements of $\R$ always define
  distinct cohomology classes. This follows from the fact that each
  summand in $d_s y
  = \sum_j (\pm\delta(e_j) \{e^*_j, y\} \pm e^*_j \{\delta(e_j), y\})$ is
  zero or has 
  a factor of nonzero degree since $\delta(e_j) = e_k$ for some $k$
  depending on $j$.
\end{proof}

\subsection{Stable Equivalence}
\label{sec:stable}

Let $P$ be a unital, Noetherian Poisson algebra and $J \subset P$ a
coisotropic ideal. Let $(X,R)$ be a BFV model for $(P,J)$ and $(Y,S)$
be a trivial BFV model. Let $\MM$ and $\NN$ be the corresponding
vector spaces. Define $Z$ as the completion of $Z_0 := X_0 \otimes Y_0
= P \otimes \Sym( \UU \oplus \UU^*)$ where $\UU = \MM \oplus \NN$ and $L = R
\otimes 1 + 1 \otimes S$. Both $X$ and $Y$
naturally sit inside $Z$ as Poisson subalgebras since the inclusions
$X_0 \to Z_0$ and $Y_0 \to Z_0$ preserve the respective filtrations.
\begin{lem}
  \label{lem:tensorstable}
  The pair $(Z,L)$ defines another BFV model for $(P,J)$.
\end{lem}
\begin{proof}
  Since the bracket between elements of $X$ and elements of $Y$ is
  zero, the element $L$ solves the master equation. The K\"unneth
	formula implies together with lemma \ref{lem:trivial-bfv-coho} 
	the conditions on the cohomology.
\end{proof}
We call $Z$ the product of $X$ and $Y$ and write $Z = X \hat{\otimes}
Y$. Adding the new variables in $\NN$ does not change the cohomology
of the BRST complex $X$:
\begin{lem}
\label{lem:quasiiso}
  The natural map $X \to Z$ defines a quasi-isomorphism of
  differential graded commutative algebras.
\end{lem}
\begin{proof}
	We define the maps $\iota : X_0 \to Z_0$ as the natural map and $p: Z_0 \to
	X_0$ as the map taking $x \otimes y$ to $x \pi(y)$, where $\pi : Y_0 \to Y_0$
	is the projection onto $\R = \Sym^0( \MM \oplus \MM^*)$ along the $\Sym^j(\MM
	\oplus \MM^*)$ with $j>0$. Both maps extend to the respective completions. We
	claim that they define mutual inverses on cohomology.

	From equation \ref{eq:ch} we infer that there exists a map $t$ on
	$Y_0$ such that $d_S t + t d_S = \id - \pi$. In particular $\pi d_S = d_S \pi
	= 0$. We have $t = \frac{1}{j} s$ on form degree $j>0$ and $t = s$ on $\R$.
	Hence $t$ preserves the filtration up to degree shift. Hence $\id \otimes t$
	extends from $Z_0$ to the completion $Z$. By tensoring the other maps too we
	obtain the identity
	\begin{align*}
		d_{R+S} (\id \otimes t) + (\id \otimes t) d_{R+S} = 
		(\id \otimes d_S)(\id \otimes t) + (\id \otimes t)(\id \otimes d_S) = \id
		\otimes (\id - \pi) = \id - \iota \circ p
	\end{align*}
	on $Z$. The first equality is true since $t$ shifts degree by one. We are
	left to show that both maps $\iota$ and $p$ descend to cohomology. For
	$\iota$ this is trivial. For $p$ note that for homogeneous $x$ of degree $k$,
	\begin{align*}
		p(d_{R+S}(x\otimes y)) = p ( (d_R x) \otimes y + (-1)^k x \otimes d_S
		y) = \pi(y) d_R x + (-1)^k x \pi d_S y = d_R( x\pi y) = d_R( p
		(x\otimes y))
	\end{align*}
\end{proof}
Now we are ready to formulate the notion of stable equivalence
introduced in \cite{felder-kazhdan}:
\begin{defi}
  Let $(X,R)$ and $(X', R')$ be two simple BFV models for $(P,J)$. We
  say that $(X,R)$ and $(X',R')$ are stably equivalent if there exist
  trivial BFV models $(Y, S)$ and $(Y', S')$ and a Poisson isomorphism
  $ X \hat{\otimes} Y \longrightarrow X' \hat{\otimes} Y'$ taking $R +
  S$ to $R' + S'$.
\end{defi}

\subsection{Relating Tate Resolutions}
\label{sec:stableuniqueness}

Now we want to consider BFV models $(R, X)$ and $(R', X')$ whose
Tate resolutions $(X/I, d_R)$ and $(X'/I',d_{R'})$ are not equal. We
have the notion of stable equivalence. Our aim is to prove that any
two such BFV models are stably equivalent and that stably equivalent
BFV models are quasi-isomorphic. As a tool we need the following lifting
statement:
\begin{lem}
  \label{lem:philift}
  Let $P = \R[V]$ with bracket induced by a symplectic structure of a
  finite-dimensional vector space $V$ and consider $T = P \otimes \Sym(\MM)$ and $T'= P \otimes
  \Sym(\MM')$. Assume there is an isomorphism $\phi : T \to T'$ of graded
  commutative algebras which is the identity in degree zero. Let $X$ be the
  completion of $X_0 = P \otimes \Sym( \MM \oplus \MM^*)$. Construct
  analogously the space $X'$. Then $\phi$ lifts to
  a Poisson isomorphism $\Phi : X \to X'$.
\end{lem}
\begin{proof}
  Since $T$ and $T'$ are negatively graded and isomorphic as graded algebras,
  we have $\MM \cong \MM'$ as graded vector spaces. Hence we may assume $\MM =
  \MM'$ and thus $T = T'$ and $X = X'$.  

  Pick standard coordinates $\{x_1, \dots, x_n, y_1, \dots, y_n\}$ on the space
  $V$ for the symplectic structure, so that $\R[V] = \R[x_i,y_j]$ and $\{x_i,
  y_j\} = \delta_{ij}$.  Let $\{e_j^{(l)}\}_j$ be a basis of $\MM^{-l}$ and
  $\{{e_j^{(l)}}^*\}_j$ be the respective dual bases. Then there are
  elements $a_{j_1 \dots j_k}^{l_1 \dots l_k}(j,l)(x_i,y_i) \in
	\R[x_i, y_i]$ and invertible matrices $a_{jk}^{(l)} \in \R[x_i,y_i]$ 
	such that
  \begin{align*}
    \phi(e_j^{(l)}) = \sum_k a_{jk}^{(l)}(x_i,y_i) e_k^{(l)} + \sum
    a_{j_1 \dots j_k}^{l_1 \dots l_k}(j,l)(x_i,y_i) e_{j_1}^{(l_1)}
    \cdots e_{j_k}^{(l_k)}
  \end{align*}
  where the sum runs over all integers $k \geq 2$ and $(j_1, l_1), \dots
  (j_k, l_k)$ with $l_1 + \cdots + l_k = l$ and is thus
  finite. Consider indeterminats $Y_i, {E_j^{(l)}}^* \in X_0$ of degree
  0 and $l$ respectively, defining
  \begin{align*}
    S(x_i, Y_i, e_j^{(l)}, {E_j^{(l)}}^*) = \sum_i x_i Y_i +
		\sum_{j,k,l} a_{jk}^{(l)}(x_i,Y_i) {E_j^{(l)}}^* e_k^{(l)} +
    \sum_{(j,l)} \sum a_{j_1 \dots j_k}^{l_1 \dots
      l_k}(j,l)(x_i,Y_i)  
			{E_j^{(l)}}^* e_{j_1}^{(l_1)} \cdots e_{j_k}^{(l_k)}
  \end{align*}
  Consider the equations
	\begin{align*}
    \frac{ \p S}{ \p x_i} & = y_i & \frac{ \p S }{ \p Y_i} & = X_i &
    \frac{\p S}{\p e_j^{(l)}} & = (-1)^l {e_j^{(l)}}^* & \frac{\p S}{\p
      {E_j^{(l)}}^*} & = E_j^{(l)}
  \end{align*}
  which read
  \begin{align*}
    y_i & = Y_i +
		\sum_{j,k,l} \frac{\p a_{jk}^{(l)}(x_i,Y_i)}{\p x_i} {E_j^{(l)}}^*
		e_k^{(l)} +
    \sum_{(j,l)} \sum \frac{\p a_{j_1 \dots j_k}^{l_1 \dots
      l_k}(j,l)(x_i,Y_i)}{\p x_i}  {E_j^{(l)}}^* e_{j_1}^{(l_1)} \cdots
    e_{j_k}^{(l_k)}\\
    X_i & = x_i+
		\sum_{j,k,l} \frac{\p a_{jk}^{(l)}(x_i,Y_i)}{\p Y_i} {E_j^{(l)}}^*
		e_k^{(l)} +
    \sum_{(j,l)} \sum \frac{ \p a_{j_1 \dots j_k}^{l_1 \dots
      l_k}(j,l)(x_i,Y_i)}{\p Y_i}  {E_j^{(l)}}^* e_{j_1}^{(l_1)} \cdots
    e_{j_k}^{(l_k)}\\
		{e_j^{(l)}}^* & =  \sum_k a_{kj}^{(l)}(x_i,Y_i) {E_k^{(l)}}^* + 
    \sum_{(j',l')} \sum a_{j_1 \dots j_k}^{l_1 \dots
		l_k}(j',l')(x_i,Y_i) (-1)^{l (l'+1)} {E_{j'}^{(l')}}^*\frac{\p (e_{j_1}^{(l_1)} \cdots
    e_{j_k}^{(l_k)})}{\p e_j^{(l)}}\\
		E_j^{(l)} & = \sum_k a_{jk}^{(l)}(x_i,Y_i) e_k^{(l)} + 
    \sum a_{j_1 \dots j_k}^{l_1 \dots
      l_k}(j,l)(x_i,Y_i)  e_{j_1}^{(l_1)} \cdots e_{j_k}^{(l_k)}
  \end{align*}
	The linear part is invertible. Hence we can solve the equations
  for $(X_i, Y_i, E_j^{(l)}, {E_j^{(l)}}^*)$ in terms of $(x_i,
  y_i, e_j^{(l)}, {e_j^{(l)}}^*)$ (and vice versa) and hence
	also for $(x_i, Y_i,
  e_j^{(l)}, {E_j^{(l)}}^*)$ in terms of $(x_i, y_i, e_j^{(l)},
  {e_j^{(l)}}^*)$ (and vice versa) in the completion $X$. Hence the
  function $S$ generates a Poisson automorphism $\Phi : X \to X$ by
  lemma \ref{lem:generatingfunction}. Let $I$ be the ideal generated by
  positive elements as defined previously. We have $\Phi(x_i) = X_i
	\equiv x_i = \phi(x_i) { \pmod I}$ and $\Phi(y_i) = Y_i
	\equiv y_i = \phi(y_i) {\pmod I}$. Thus also
	$\Phi(e_j^{(l)}) = E_j^{(l)} \equiv \phi(e_j^{(l)}) {\pmod
	I}$.  Hence $\Phi$ is a lift of $\phi$.  
\end{proof}

\begin{thm}
  \label{thm:stableequivalence}
  Consider $P = \R[V]$ with bracket induced by a symplectic structure on a
  finite-dimensional vector space $V$. Any two BFV models for $(P,J)$ are
  stably equivalent.  \end{thm}
\begin{proof}
  Let $(X,R)$ and $(X',R')$ be BFV models with associated Tate
  resolutions $T := X/I \cong P \otimes \Sym(\MM)$ and $T' := X'/I'
  \cong P \otimes \Sym(\MM')$. By~\cite[theorem A.2]{felder-kazhdan},
  there exist negatively graded vector spaces $\NN$ and $\NN'$ with
  finite dimensional homogeneous components, differentials
  $\delta_{\NN} : \Sym(\NN) \to \Sym(\NN), \delta_{\NN'} : \Sym(\NN')
  \to \Sym(\NN')$ with cohomology $\R$, and an isomorphism $\phi$ of
  differential graded commutative algebras
  \begin{align*}
    P\otimes \Sym( \MM \oplus \NN ) \to & P\otimes \Sym(\MM' \oplus
    \NN')
  \end{align*}
  restricting to $\id_P : P \to P$ in degree 0.  Let $Y$ and $Y'$ be
  the trivial BFV models corresponding to $\NN$ and $\NN'$ with BRST
  charges $S$ and $S'$, respectively. Consider the spaces $Z = X
  \hat{\otimes} Y$ and $Z' = X' \hat{\otimes} Y'$. Together with the
  operators $L = R + S$ and $L' = R' + S'$ they form BFV models $(Z,
  L)$ and $(Z', L')$ for $(P,J)$ by lemma \ref{lem:tensorstable}. 
  
  We now construct a Poisson isomorphism $\Phi : X \hat{\otimes} Y \to
  X' \hat{\otimes} Y'$ sending $R + S$ to $R' +
  S'$. By lemma \ref{lem:philift}, the map $\phi$ lifts to a Poisson
  isomorphism  $\Psi : X \hat{\otimes} Y \to X' \hat{\otimes}
  Y'$. Now $L''= \Psi(L)$ solves
  $\{-,-\}=0$ in $X' \hat{\otimes} Y'$. Moreover, $\{L'', -\}$ induces
  $\delta'$ on $P \otimes \Sym( \MM' \oplus \NN')$. By theorem
  \ref{thm:gaugeequivalence}, there exists a Poisson isomorphism
  $\chi$ of $X' \hat{\otimes} Y'$ with $L' = \chi( L '')$. Set $\Phi =
  \chi \circ \Psi$.

  We are now in the situation
  \begin{align*}
    \xymatrix{ X \ar[d] & X' \ar[d] \\
      X \hat{\otimes} Y \ar[r]^{\Phi} & X' \hat{\otimes} Y'}
  \end{align*}
  where the vertical arrows represent natural maps which are
  quasi-isomorphisms by lemma \ref{lem:quasiiso}.
\end{proof}

\begin{lem}
  \label{lem:stabcoho}
  The complexes of two stably equivalent BFV models are
  quasi-isomorphic. In particular, they have cohomologies which are
  isomorphic as graded commutative algebras.
\end{lem}
\begin{proof}
  Let $(X,R)$ and $(X',R')$ be two stably equivalent BFV
  models. Hence we are in the situation
  \begin{align*}
    \xymatrix{ X \ar[d] & X' \ar[d] \\
      X \hat{\otimes} Y \ar[r] & X' \hat{\otimes} Y'}
  \end{align*}
  where the downward arrows are quasi-isomorphisms of differential
  graded commutative algebras by lemma \ref{lem:quasiiso} and the
  bottom arrow is a Poisson isomorphism $X \hat{\otimes} Y \to X'
  \hat{\otimes} Y'$ sending $R + S$ to $R' + S'$.
\end{proof}
From theorem \ref{thm:stableequivalence} and lemma \ref{lem:stabcoho}
we obtain analogously to the treatment of the Lagrangian case in
\cite{felder-kazhdan} 
\begin{cor}
  Let $P=\R[V]$ with bracket induced by a symplectic structure on a
  finite-dimensional vector space $V$. 
  Any two BRST-complexes arising from BFV-models for the same coisotropic
  ideal $J\subset P$ are quasi-isomorphic. Hence, the BRST cohomology
  is uniquely determined by $(P=R[x_i,y_i],J)$ up to an isomorphism of
  graded commutative algebras.
\end{cor}

\section{Cohomology}
\label{sec:cohomology}

Let $P$ be a unital, Noetherian Poisson algebra and $J$ a coisotropic
ideal. Let $(X,R)$ be a BFV model for $J \subset P$. In this section
we analyze the cohomology of the complex $(X,d_R)$. We follow the
strategy from \cite{felder-kazhdan}.

\subsection{Cohomology and Filtration}
The associated graded is defined as $\gr^p X = \FF^pX /
\FF^{p+1} X$. The differential $d_R$ induces a map $\delta$ on
$X/I = T = P \otimes \Sym(\MM)$ and the results from section
\ref{sec:properties} apply.
\begin{lem}
    \label{lem:cohocalc}
    $ H^j( \gr^p , d_R )\cong  B^p \otimes_P P/J$ for $j=p$ and $ H^j(
    \gr^p , d_R ) \cong 0$ for $j \neq p$.
\end{lem}
\begin{proof}
  Fix $p$. $B^p$ is a free $P$-module. By lemma \ref{lem:induceddiff}, we have
  \begin{align*}
    H^j( \FF^p X / \FF^{p+1} X, d_R )\cong
    H^j( B^p \otimes_P T^{\bullet - p} , 1 \otimes \delta)
    \cong B^p \otimes_P H^{j-p}(T, \delta) \cong
      B^p \otimes_P H^{j-p}(X/I, d_R)
  \end{align*}
\end{proof}
Next, we want to prove that, in order to compute the cohomology in a
fixed degree, one may disregard elements of high filtration degree.
\begin{lem}
  \label{lem:fpbecomeszero}
  Let $j < p$ be integers with $p \geqslant 0$. Then $H^j( \FF^p X,
  d_R ) = 0$. 
\end{lem}
\begin{proof}
  Let $x \in \FF^p X^j$ be a cocycle representing a cohomology class
  in $H^j(\FF^p X, d_R)$. Then $x + \FF^{p+1} X^j$ defines a cocycle
  in $H^j(\FF^p X/ \FF^{p+1} X, d_R)$. By lemma \ref{lem:cohocalc},
  there is $y_0 \in \FF^p X^{j-1}$ with $x - d_Ry_0 \in \FF^{p+1}
  X^j$. Hence this element defines a cocycle in $H^j(\FF^{p+1} X /
  \FF^{p+2} X, d_R) = 0$. Hence there is $y_1 \in \FF^{p+1} X^{j-1}$ with
  $x - d_Ry_0 - d_Ry_1 \in \FF^{p+2}X^j$.  Iterating this procedure we
  find a sequence $y_0, y_1, \dots$ of elements $y_j \in \FF^{p+j}
  X^{j-1}$ with $x - d_R(y_0 + \cdots + y_j) \in \FF^{j+1} X^j$. By
  lemma \ref{lem:filtrationsums} the element $y := y_0 + \cdots \in
  X^{j-1}$ is well-defined and $y_0 + \cdots + y_j \to y$. Since all
  $y_j$ are in $\FF^p X^{j-1}$ and this set is closed by lemma
  \ref{lem:filtrationclosed}, we have $y \in \FF^p X^{j-1}$. Finally,
  for $n$ fixed, and all $j$,
  \begin{align*}
    d_R y_0 + \cdots + d_R y_n + \cdots + d_R y_{n+j} - x \in
    \FF^{n+1} X^j
  \end{align*}
  Since $d_R = \{R, -\}$ is continuous (lemma
  \ref{lem:bracketcontinuous}), we have $d_R y - x \in \FF^{n+1}
  X^j$. Since $n$ was arbitrary, $d_R y = x$.
\end{proof}
\begin{cor}
\label{cor:cohononneq}
  The cohomology of $(X, d_R)$ is concentrated in non-negative degree.
\end{cor}
\begin{cor}
  \label{cor:finitecohomology}
  The natural map $H^j( X, d_R ) \to H^j( X / \FF^{p+1} X, d_R)$ is an
  isomorphism for $j < p$ and injective for $j=p$. 
\end{cor}
\begin{proof}
  The short exact sequence $0 \to \FF^{p+1} X \to X \to X / \FF^{p+1}
  X \to 0$ defines the long exact sequence
  \begin{align*}
    \cdots \to H^j( \FF^{p+1}X, d_R ) \to H^j( X, d_R ) \to H^j ( X / \FF^{p+1}
    X, d_R ) \to H^{j+1}( \FF^{p+1} X, d_R ) \to \cdots
  \end{align*}
  For $j \leqslant p$ the first term is zero and for $j < p$ both the
  first and the last terms are zero by lemma \ref{lem:fpbecomeszero}.
\end{proof}

\subsection{Spectral Sequences}

\begin{lem}
  \label{lem:ssdeg}
  Let $E^{p,q}_r$ be the spectral sequence corresponding to the
  filtered complex $\FF^p X^{p+q}$ with differential $d_R$.  We have
  $H^\bullet( X, d_R ) \cong E^{\bullet,0}_2$ as graded commutative
  algebras.
\end{lem}
\begin{proof}
  Begin with $E_0^{p,q} := \FF^p X^{p+q} / \FF^{p+1} X^{p+q}$. It is
  concentrated in degree $p \geqslant 0, q \leqslant 0$. By lemma
  \ref{lem:cohocalc} we have the following isomorphism of differential
  bi-graded algebras:
  \begin{align*}
    E^{p,q}_1 = H^q ( E_0^{p, \bullet} , d_R ) = H^q( \FF^p X^{p+ \bullet} /
    \FF^{p+1} X^{p + \bullet}, d_R )&= H^{p+q}(\FF^p X / \FF^{p+1}
    X, d_R) \\
    &\cong
    \begin{cases}
      B^p \otimes_P P/J, & \text{if $q=0$} \\\nonumber 0, & \text{if
        $q\neq 0$}.  \end{cases}
  \end{align*}
  Hence $E^{p,q}_1$ is concentrated in degree $p \geqslant 0$ and $q =
  0$.  Moreover, $d_1^{p,q}$ maps $E_1^{p,q}$ to $E_1^{p+1, q}$. Hence
  also $E_2^{p,q}$ is concentrated in $p \geqslant 0, q = 0$. Since
  $d_2$ maps $E^{p,0}_2 $ to $E_2^{p+2,-1}$ it is zero for degree
  reasons and hence the spectral sequence degenerates at $E_2$.
  
  We are left to prove that the spectral sequence converges to the
  cohomology. By \cite[chapter XV, proposition
  4.1]{cartan-eilenberg1956}, this follows from lemma \ref{lem:fpbecomeszero}.
\end{proof}
We could use this lemma to prove $H^0(X, d_R) \cong (P/J)^J$ as
algebras. However, we want to consider an additional structure on the
latter space.

\subsection{The Poisson Algebra Structure on $(P/J)^J$}

We have a Poisson algebra structure on $(P/J)^J$.
\begin{lem}
  The Poisson algebra structure on $P$ induces a Poisson algebra
  structure on $(P/J)^J$.
\end{lem}
\begin{proof}
  Let $p + J, q + J \in (P/J)^J$, i.e. $p,q \in P$ with $\{p, J\}, \{q,
  J\} \subset J$ and $a \in \R$. $J$ is a linear subspace of $P$ and
  $\{a p + q, J \} \subset J$ hence $(P/J)^J$ is a vector space. $J$ is
  an ideal in $P$ and $\{pq, J\} \subset J$ by the Leibnitz rule. Hence
  $(P/J)^J$ is an algebra. Finally, we have to show that the Poisson
  bracket descends to $(P/J)^J$. We have $\{p, q\} \in J$ if $p \in
  J$. Hence the definition $\{p+J, q+J\} := \{p,q\} + J$ does not depend
  on the choice of representatives. Moreover, $\{\{p,q\}, J\} \subset J$ by
  the Jacobi identity. Hence the bracket is well-defined.
\end{proof}
We also have a Poisson algebra structure on $H^0(X, d_R)$. 
\begin{lem}
  The graded Poisson algebra structure on $X$ induces a Poisson
  algebra structure on the cohomology $H^0(X, d_R)$ in degree zero. 
\end{lem}
\begin{proof}
  The cohomology of a differential graded commutative algebra is
  naturally a graded commutative algebra. In particular, the
  cohomology in degree 0 is a commutative algebra. We have to show
  that the bracket descends to $H^0(X, d_R)$. Let $x, y \in X^0$
  be representatives of cohomology classes in $H^0(X, d_R)$. Let $a
  \in \R$. Then $\{x,y\} \in X^0$ is closed: $\{R, \{x,y\}\} = 0$ by the
  graded Jacobi identity. Moreover, if $x = d_R x'$ is exact, then
  $d_R \{x', y\} = \{R, \{x', y\}\} = -\{x', \{y, R\}\} - \{y, \{R, x'\}\} = -\{y, x\}
  = \{x,y\}$.
\end{proof}
Those two structures are in fact isomorphic. We will explicitly
construct a Poisson isomorphism. By corollary \ref{cor:finitecohomology},
we have $H^0(X, d_R) \cong H^0(X/\FF^2 X, d_R)$ as vector spaces.
\begin{lem}
  \label{lem:repcoho}
  Representatives in $X^0$ of cocycles in $X^0/\FF^2 X^0$ defining
  elements in $H^0(X/\FF^2 X , d_R)$ may be taken of the form
  \begin{align*}
    x = x_0 + \sum_{i,j \in L} a_{ij} e^*_i e_j
  \end{align*}
  where $L = \{ n \in \N : \deg(e^*_j) = 1 \}$, $x_0 \in P$, the $\{e_j\}$ are
  a homogeneous basis of $\MM$, and the $a_{ij} \in P$ are chosen such that 
  \begin{align*}
    \{\delta(e_j), x_0\} = \sum_{i \in L} a_{ji} \delta(e_i)
  \end{align*}
  Conversely, every such element defines a cohomology class.
\end{lem}
\begin{proof}
  We have
  \begin{align*}
    X^0 / \FF^2 X^0 &= P \oplus ( P \otimes (\MM^*)^{-1} \otimes
    \MM^{-1} ) \\
    X^{-1} / \FF^2 X^{-1} & = (P \otimes \MM^{-1}) \oplus ( P \otimes (
    \MM^*)^1 \otimes \MM^{-2}) \oplus (P \otimes (\MM^*)^1 \otimes
    (\MM^{-1} \wedge \MM^{-1})).
  \end{align*}
  Hence an arbitrary cochain may be taken to be of the form 
  \begin{align*}
    x = x_0 + \sum_{i,j \in L} a_{ij} e^*_i e_j
  \end{align*}
  for some $x_0, a_{ij} \in P$.  We compute with the help of lemma
  \ref{lem:q0},
  \begin{align*}
    d_R x &= \{R, x_0\} + \sum_{i,j \in L} ( \{R, a_{ij} \}e^*_i e_j + \{R,
    e^*_i \} e_j a_{ij} - \{R,  e_j\} e^*_i a_{ij}) \equiv \{R, x_0\} -
    \sum_{i,j \in L} \{R, e_i\} e_j^* a_{ji}\\
    &\equiv  \sum_{j \in L}
    \bigg( (-1)^{1+d_j}\{\delta(e_j), x_0\} - \sum_{i \in L} a_{ji} \delta(e_i)
    \bigg) e^*_j \pmod{ \FF^2 X^1}
  \end{align*}
\end{proof}
\begin{thm}
  $H^0(X, d_R) \cong (P/J)^J$ as Poisson algebras.
\end{thm}
\begin{proof} 
  Let $\pi : X \to P = X/(I + I_-)$ denote the projection onto all monomials which
  contain no factors of nonzero degree. Here $I_- \subset X$ denotes the ideal
  generated by all elements of negative degree. Define the map $\Phi : H^0(X,
  d_R) \to P/J$ by $\Phi([x]) := \pi(x) + J$.  This map is well
  defined: Let $x = d_R y$ be exact. Consider again the differential $\delta$
  on $T=X/I$ that is induced by $d_R$ and its representation as an element $Q_0
  \in X^1$. Also pick a homogeneous basis $e_j$ of
  $\MM$ as done before. By lemma \ref{lem:q0}, we obtain 
  \begin{align*}
	 d_R(y) \equiv \{Q_0, y\} \equiv \sum_i \{ (-1)^{1+d_i} \delta(e_i) e_i^*,
  y\} \equiv \sum_{i : \deg e_i = -1} \delta(e_i) \{ e_i^*, y\} \pmod{ I +
  I_-}.
  \end{align*}
  The last sum is finite. Hence
  \begin{align*}
	 \pi(x) = \pi(d_R y) = \sum_{i : \deg e_i = -1} \delta(e_i) \pi \{e_i^*, y\}
	 \in J.
  \end{align*}

  By lemma \ref{cor:finitecohomology}, we have $H^0(X, d_R) \cong
  H^0(X/\FF^2 X, d_R)$ as vector spaces. Hence we have a corresponding
  linear map $X/\FF^2 X \to P/J$.

  The image of either of those maps is $J$-invariant: Let $[x] \in
  H^0( X / \FF^2 X, d_R)$. According to lemma \ref{lem:repcoho}, we
  may pick a representative $x_0 = \pi(x_0) + \sum_{i,j \in L} a_{ij}
  e^*_i e_j$ of $x$ where $a_{ij} \in P$ satisfy $\{\delta(e_j),
  \pi(x_0) \} = \sum_{i\in L} a_{ji} \delta(e_i)$. In particular
  $\{\delta(e_j), \pi(x_0)\} \in J$. Fix $b \in J$. Then there exist
  $b_j \in P$ with $b = \sum_{i \in L} b_i \delta(e_i)$ and thus $ \{b,
  \pi(x_0)\} = \sum_{i \in L} \big (b_i \{\delta(e_i), \pi(x_0)\} +
  \delta(e_i) \{b_i, \pi(x)\} \big) \in J$.

  Hence we have two linear maps
  \begin{align*}
    \phi : H^0(X/\FF^2 X, d_R) &\to (P/J)^J \\
    \Phi : H^0(X, d_R) &\to (P/J)^J
  \end{align*}
  given by projection onto the $P$ component followed by modding out
  $J$, which correspond to each other under the isomorphism $H^0(X,
  d_R) \cong H^0(X/\FF^2 X, d_R)$.

  The map $\phi$ is surjective: Let $p \in P$ with $\{J,p\}\subset
  J$. By lemma \ref{lem:repcoho}, the element $x = p + \sum_{ij
    \in L} a_{ij} e^*_i e_j$ is a cocycle if $\{\delta(e_j), p\} =
  \sum_{i\in L} a_{ji} \delta(e_i)$. But those $a_{ij} \in P$ exist
  since the $\{\delta(e_j)\}_{ j \in L}$ generate $J$. Hence also the
  map $\Phi$ is surjective.

  The map $\Phi$ is injective: Let $x \in X^0$ represent $[x] \in
  H^0(X, d_R)$ with $\pi(x) \in J$. We claim that there exist $y_j \in
  \FF^j X^{-1}$ with $x - d_R( y_0 + \cdots + y_n ) \in \FF^{n+1}
  X^0$. By lemma \ref{lem:cohocalc}, we know that $ H^j( \FF^p X /
  \FF^{p+1} X, d_R )$ is concentrated in degree zero with $H^0(X
  / \FF^{1} X, d_R ) \cong P/J$ via the natural map. Now $x + \FF^1
  X^0$ defines the zero cohomology class in $H^0( X / \FF^1 X, d_R)$
  since $\pi(x) \in J$. Hence there exists $y_0 \in \FF^0 X^{-1}$ with
  $x - d_R y_0 \in \FF^1 X^0$. Again, $x - d_R y_0 + \FF^2 X^0$
  defines the zero cohomology class in $H^0( \FF^1 X /\FF^2 X, d_R) =
  0$. Hence there exists $y_1 \in \FF^1 X^{-1}$ with $x - d_R(y_0 +
  y_1) \in \FF^2 X^0$ and so on. Hence the $y_j$ exist and their sum
  converges to an element $y \in X^{-1}$ by lemma
  \ref{lem:filtrationsums}, which satisfies $x - d_R y = 0$ by lemma
  \ref{lem:filtrationclosed}. 
  
  Hence the map $\Phi$ is an isomorphism of vector spaces.
  This map also respects the product structure
  \begin{align*}
    \Phi( [x] [y] ) & = \Phi( [xy] ) = \pi(x y) + J = \pi(x) \pi(y) +
    J = (\pi(x) + J)(\pi(y) + J)= \Phi([x])\Phi([y])
  \end{align*}
  and is hence an isomorphism of algebras. Finally, map $\Phi$
  respects the bracket:
  \begin{align*}
    \Phi( \{ [x], [y] \}) &= \Phi( [ \{x,y\}]) = \pi(\{x,y\}) + J =
    \pi(\{\pi(x), y\} ) + J = \pi( \{\pi(x), \pi(y)\}) + J \\ &= \{\pi(x),
    \pi(y)\} + J = \{\pi(x) + J, \pi(y) + J\} = \{\Phi(x), \Phi(y)\}
  \end{align*}
  since $\{ \pi(x) - x , X^0\} \subset \{ (I + I_0) \cap X^0, X^0\} \subset
  \ker \pi$ where $\ker \pi = I + I_0 \subset X$ is the ideal generated by all
  elements of nonzero degree. The last inclusion holds by the Leibnitz rule
  since all summands of elements in $I + I_0$ that are of degree zero contain
  at least two factors of nonzero degree.
\end{proof}

\section{Examples}

\subsection{Rotations of the Plane}
\label{sec:sphere}

Here we present an example, where the cohomology in degree zero has a
nontrivial bracket and the cohomology in degree 1 does not vanish. It
is obtained by considering the symplectic lift of the rotations of the
plane to the cotangent bundle of the plane.

Consider $P = \R[x_1, x_2, y_1, y_2]$ with $\{x_i, y_j\} =
\delta_{ij}$. The ideal $J \subset P$ generated by $\mu = x_1 y_2 -
x_2 y_1$ is coisotropic.  A Tate resolution of $J$ is given by
\begin{align*}
  0 \to P \cdot e \to P \to P/J \to 0
\end{align*}
where the differential $\delta$ is the $P$-linear derivation defined by
$\delta(e) = \mu$. Indeed, this complex is a Koszul complex
which is exact since $\mu\neq 0$ defines a regular sequence. Hence $X = \big(
P \cdot e \big) \oplus \big( P \oplus P \cdot e^* e \big) \oplus \big(
P \cdot e^* \big)$.  We now
apply the construction from section \ref{sec:existence}.  We obtain
$Q_0 = e^* \mu$ and $R = Q_0$ since $\{Q_0, Q_0\} = 0$. One easily calculates
\begin{align*}
  H^0(X, d_R) = \frac{ \{ a + b e^*e : \{\mu, a\} = \mu b , a,b \in P\}
  }{ \{ \mu c + \{\mu,c\} e^* e : c \in P \} }.
\end{align*}
Notice, that the isomorphism $H^0(X, d_R) \to (P/J)^J$ given by
projection onto $P$ is evident here. Moreover, the bracket on this
space does not vanish: $x_1^2 + x_2^2$ and $y_1^2 + y_2^2$ define
cohomology classes, for which $\{x_1^2 + x_2^2, y_1^2 + y_2^2\} = 4 (x_1 y_1
+ x_2 y_2)$ is not in $J$.  Furthermore,
\begin{align*}
  H^1(X, d_R) = \frac{\{ a e^* : a \in P\}}{\{ d_R( a + b e^*e ) : a,
    b \in P\}} \cong \frac{ P }{ \{\{\mu, a\} + \mu b  : a, b \in P \}}
\end{align*}
does not vanish since $\deg_0\{\mu, a\} \geq 1$ and $\deg_0(\mu b) \geq
2$. Here, $\deg_0$ denotes the degree in $P=\R[x_i, y_i]$.

\subsection{Rotations of Space}
\label{sec:space}

Let $X = \R^3$ and $M = T^*X \cong X \oplus X^*$. Consider the group $G =
SO(3)$ acting on $X$ via the standard representation $\rho_0 : G \to \End{X}$.
The symplectic lift is given by $\rho : G \to \End{M}$, $\rho(A)(x,p) = (Ax, p
\circ A^{-1})$. Mapping the standard basis of $X = \R^3$ to its dual basis, we
obtain an isomorphism $\iota : X \to X^*$. A possible moment map is the angular
momentum mapping $\mu : M \to \R^3$, $\mu(x,p) = x \times \iota{p}$. Here
$\times$ refers to the vector product, and we identified $\gfrak \cong \R^3$
using the basis
\begin{align*}
		  e_1 & = 
		  \begin{pmatrix}
					 0 & 0 & 0 \\
					 0 & 0 & -1 \\
					 0 & 1 & 0
		  \end{pmatrix},
		  &
		  e_2 & =
		  \begin{pmatrix}
					 0 & 0 & -1 \\
					 0 & 0 & 0 \\
					 1 & 0 & 0
		  \end{pmatrix},
		  &
		  e_3 & = 
		  \begin{pmatrix}
					 0 & -1 & 0 \\
					 1 & 0 & 0 \\
					 0 & 0 & 0
		  \end{pmatrix}.
\end{align*}
We define $M_0 = \mu^{-1}(0) = \{ (x,p) \in X \oplus X^* : \iota p \parallel
x\}$. This is not a manifold. If it was one, it had dimension 4. However, all
$(x,0), x \in \R^3$ and all $(0,p), p \in (\R^3)^*$ belong to $M_0$. Hence they
also belong to the tangent space at the origin, provided $M_0$ was a manifold.
Since the tangent space at the origin is linear, it would have dimension $6 > 4$. Since the
constraint surface $M_0$ is not a manifold, results from \cite{fisch1989},
\cite{kostant-sternberg} do not apply.

We take the Tate resolution $T$ of the vanishing ideal $J$ of
$M_0$ in $P = \R[ M ] = \R[x_1,x_2,x_3,p_1,p_2,p_3]$ with $\{x_i,p_j\} = \delta
_{ij}$. It exists since $P$ is Noetherian. We obtain the existence and
uniqueness of a BRST charge $R$ as described in the previous chapters.

%
%
%

\section{Quantization}

In this section, we discuss quantization. In section \ref{ssec:qalg}, we define
a quantum algebra quantizing the Poisson algebra $X$ from the previous part of
this note. We rigorously define multiplication via normal ordering in the
presence of infinitely many ghost variables. 

In section \ref{ssec:qmesolution}, we construct a solution of the quantum
master equation associated to a given solution of the classical master
equation. This means we construct an element of the quantum algebra that agrees
with the quantization of the classical solution up to an error of order $\hbar$
and squares to zero.

In section \ref{ssec:qmeunique}, we discuss the uniqueness of such solutions of
the quantum master equation. We parallel our discussion in the classical case.
In section \ref{sssec:qgauge}, we define the
notion of a quantum gauge equivalence. In section \ref{sssec:uniquenessgivenR},
we prove that two solutions of the quantum master equation that agree up to an
error of order $\hbar$ are related via an automorphism of associative algebras.
In section \ref{sssec:qsametate}, we show that the solutions of the quantum
master equation associated to two BRST-models associated to the same Tate
Resolution are also related by an automorphism of associative algebras. In
section \ref{sec:stableuniqueness}, we have shown that any two BRST-models
associated to the same coisotropic ideal $J$ are stably equivalent. We would
like to find a quantum analogue of this theorem. We were able to prove in
section \ref{sssec:quantproduct} that the process of adding extra variables
yields a quasi-isomorphism of differential graded algebras on the quantum
level. However, as discussed in section \ref{sssec:genunique}, we were unable
to quantize the general Poisson isomorphism of lemma \ref{lem:philift}.

\subsection{Quantum Algebra}
\label{ssec:qalg}

\begin{ass}
  Assume that $P = \R[V]$ where $V$ is a finite-dimensional real symplectic
  vector space and the bracket on $P$ is induced by the symplectic structure. 
\end{ass}
Pick a decomposition $V = W \oplus W^*$ where $W$ is a Lagrangian subspace of
$V$ and $\{v, \lambda\} = \lambda(v)$ for $v \in W$ and $\lambda \in W^*$. 
Set $\NN := \MM \oplus W$ so that we may write $X_0 = \Sym( \NN \oplus \NN^*)$.
Further define $X_- = \Sym(\NN)$ and $X_+ = \Sym(\NN^*)$. The first
approximation to the quantum algebra is the vector space $G_0 = G_- \otimes
G_+$. Here, $G_\pm = X_\pm$ as vector spaces. We discard the algebra structure
induced by the tensor product. Later we will construct a different one. The
canonical isomorphism $\Sym(\NN) \otimes \Sym(\NN^*) \to \Sym(\NN \oplus
\NN^*)$ has an inverse which we call the classical normal ordering map $q_0 :
X_0 \to G_0$. It is an isomorphism of graded vector spaces. Although, given
$W$, the map $q_0$ is canonical, the choice of $W$ is not.

Our first goal is to define a suitable algebra structure to realize the quantum
commutation relations. For this task, we need to introduce a formal parameter
$\hbar$ and construct the graded vector space $G_0[\hbar]$. We consider $\hbar$
to be of degree zero. Our second
goal is to complete this vector space in such a way that we can extend the
normal ordering map $q_0 : X_0 \to G_0 \subset G_0[\hbar]$ to a quantization
map defined on all of $X$.

First, we introduce normal ordering on the quantum level. Let $s_- :
\Sym(\NN) \to T(\NN)$ and $s_+ : \Sym(\NN^*) \to T(\NN^*)$ be any two
$\R$-linear splits from the symmetric algebras to the respective tensor
algebras. Denote the canonical isomorphism $T(\NN) \otimes T(\NN^*) \to T(\NN
\oplus \NN^*)$ by $\tau$ and set $s = \tau \circ (s_- \otimes s_+) :
G_0 \to T(\NN \oplus \NN^*)$. Extend the latter map to 
\begin{align*}
s : G_0[\hbar] \to T(\NN
\oplus \NN^*)[\hbar]
\end{align*}
by linearity in $\R[\hbar]$. Similarly, using the canonical projections $T(\NN)
\to \Sym(\NN)$ and $T(\NN^*) \to \Sym(\NN^*)$, define the map $
	(T(\NN) \otimes T(\NN^*))[\hbar] \to G_0[\hbar]$
by $\R[\hbar]$-linear extension. Let $(cr) \subset T(\NN
\oplus \NN^*)[\hbar]$ be the two-sided ideal generated by 
\begin{align*} [x,y] - \hbar \{x,y\}, \quad x,y \in \NN \oplus
  \NN^* 
\end{align*} 
where $[-,-]$ denotes the graded commutator on the tensor algebra. To define a
multiplication on $G_0[\hbar]$ we have to prove the following Poincar\'e-Birkhoff-Witt
type statement:
\begin{thm}
	Let $x , y \in G_0[\hbar]$.  Then the class of 
  $s(x) s(y) \in T(\NN \oplus \NN^*)[\hbar]$ modulo $(cr)$ has a
  representative in $(T(\NN) \otimes
	T(\NN^*)) [\hbar]$. Its image in $G_0[\hbar]$ is uniquely determined by $x$ and $y$
  and is moreover independent of the choice of splits $s_{\pm}$. We denote it by
	$\nu(x,y)$. The resulting map $\nu : G_0[\hbar] \times G_0[\hbar] \to G_0[\hbar]$ is
  $\R[\hbar]$-bilinear.
\end{thm}
\begin{proof}
Pick an ordered homogeneous basis $\{\xi_i\}$ of $\NN$ and its
corresponding dual basis $\{ \eta_i \}$ of $\NN^*$, so that $G_0 = \R[\xi_i]
\otimes \R[\eta_i]$. We now have $s(x)s(y) \in F_0[\hbar]$ where $F_0 = \R\la
\xi_i, \eta_i \ra$ is a free algebra. By theorem \ref{thm:pbw}, we obtain a
unique normal ordered representative of $s(x)s(y) + (cr) \in F_0[\hbar]/(cr)$,
which lives in $(\R\la \xi_i \ra \otimes \R \la \eta_i \ra )[\hbar]$. Its image
in $(\R[\xi_i]\otimes\R[\eta_i])[\hbar]$ is denoted by $\nu(x,y)$. If we change
either of the splits $s_\pm$ the class $s(x)s(y) + (cr)$ will not change, since
the difference lives in the ideal $(cr)$. Moreover, this class depends
$\R[\hbar]$-linearly on $x$ and $y$.  By theorem \ref{thm:pbw}, also the normal
ordered representative depends linearly on that class. Hence $\nu(x,y)$ is
$\R[\hbar]$-bilinear. By lemma \ref{lem:pbw-base-indep}, the element $\nu(x,y)$
does not depend on the choice of basis $\{\xi_i\}$.
\end{proof}
\begin{lem}
	$(G_0[\hbar], \nu)$ is a graded unital associative algebra.
\end{lem}
\begin{proof}
		  We have to show associativity. We have $s( \nu(x,y) ) \equiv s(x) s(y)
		  \pmod{ (cr)} $ by construction. Hence 
		  \begin{align*}
					 s( \nu(x,y) ) s(z) \equiv s(x) s(y) s(z) \equiv s(x)
					 s(\nu(y,z)) \pmod{ (cr)}
		  \end{align*}
		  and associativity follows.
\end{proof}
We write $ab$ instead of $\nu(a,b)$.
Define $[a, b] = ab - (-1)^{\deg a \deg b} ba$.
\begin{rem}
  If we pick a basis $\{x_i,y_i\}$ of $V$ adapted to the symplectic structure
  and a homogeneous basis $e_i$ of $\MM$ and corresponding dual bases $y_i \in
  V^*$ and $e_i^* \in \MM^*$ we have
	\begin{align*}
	  [x_i, y_j] &= \hbar \delta_{ij} & 
	  [e_k, e_l^*] &= \hbar \delta_{kl} \\
	  [x_i, x_j] &= 0 &  [y_i,y_j] &= 0 & [e_k, e_l] &= 0 &  [e_k^*,
	  e_l^*] = 0 \\ 
	  [x_i, e_k] & = 0 &  [x_i, e_k^*] &= 0 & [y_i, e_k] &= 0 & [y_i,
	  e_k^*] = 0. 
	\end{align*}
\end{rem}
Hence we have realized the quantum commutation relations. 

Next we define a completion and extend the product to it.
We introduce the filtration on $G_0$ defined by the subspaces $\FF^p G_0 =
\bigoplus_{q \geq p} G_- \otimes G_+^p$. Set $\FF^p G_0^n = \FF^p G_0 \cap
G_0^n = \bigoplus_{q \geq p} G_-^{n-p} \otimes G_+^p$. We complete $G_0$ to the
graded vector space $G = \bigoplus_n G^n$ where
\begin{align*}
  G^n = \lim_{\leftarrow p} \frac{G_0^n}{\FF^p G_0^n}.
\end{align*}
This graded vector space is again filtered by the subspaces $\FF^p G^n =
\lim_{\leftarrow q} \frac{\FF^p G_0^n}{\FF^{p+q} G_0^n}$.

Define the graded vector space $G_\hbar$ by its homogeneous components $H_\hbar^n =
G^n[[\hbar]]$. We have a family of projections 
 $p_j : G_\hbar \to G$ mapping $\sum_{k \geq 0} x_k \hbar^k \mapsto x_j$.
 The space $G_0[\hbar]$ can be considered a graded subspace of $G_\hbar$. We
 want to extend the algebra structure on $G_0[\hbar]$ to $G_\hbar$. For this
 task, we need to analyze the compatibility of the product on $G_0[\hbar]$ with the
 filtration on $G_0$.
\begin{lem}
  \label{lem:multcont}
	We have for all $j, p \geq 0$ and $n,m \in \Z$
	\begin{enumerate}
	  \item $p_j( G_0^n \cdot \FF^p G_0^m ) \subset \FF^p G_0^{n+m}$
	  \item $p_j( \FF^p G_0^n \cdot G_0^m ) \subset \FF^{p+m} G_0^{n+m}$
	\end{enumerate}
\end{lem}
\begin{proof}
  Consider $x = a\otimes u \cdot b\otimes v \in G_0^n \cdot G_0^m$ where 
	$ a \in G_-^{n-l}, u \in G_+^l, b \in G_-^{m-k}$ and $v  \in G_+^k$
  	for some $l, k \geq 0$. Then
	\begin{align*}
	  p_j(a\otimes u \cdot b\otimes v) = \sum_{b' u'} \pm ab' \otimes u'v
	\end{align*}
	where $u'$ and $b'$ arise from $u$ and $b$ by deleting $j$ matching pairs
	$(e,e^*)$ where $e \in \NN$ is a factor in $b$ and
	$e^* \in \NN^*$ is a factor in $u$ of opposite degree. The sum is finite.

	To prove the first statement, it suffices to note that $\deg u' \geq 0$ and
	$k \geq p$. For the second statement we note that $\deg u' + \deg
	b' = \deg u + \deg b$ and hence we can estimate $ \deg u' + \deg v \geq \deg
	u' + \deg b' + \deg v = \deg u + \deg b + \deg v = l + m-k + k \geq p +
	m$.
\end{proof}
\begin{lem}
	The product (and hence also the commutator) extend to the completion
 turning $(G_\hbar, \cdot)$ into a graded algebra. \end{lem} 
\begin{proof}
  First we consider $x = (x_p + \FF^p G_0^n)_p \in G^n$ and $y = (y_p + \FF^p
  G_0^m)_p \in G^m$. By lemma \ref{lem:multcont} the limit $\lim_{p \to \infty}
  x_p \cdot y_p \in H^{n+m}$ is well-defined and the definition $x \cdot p =
  \lim_{p\to \infty} x_p \cdot y_p$ does not depend on the choice of
	representatives $x_p, y_p$. The multiplication extends to $G_\hbar$ by
	bi-linearity in $\R[ [\hbar]]$.
\end{proof}
\begin{rem}
  \label{rem:explicitproduct}
  More explicitly, we have $x\cdot y = \sum_j (
  p_j (x_{p-m} \cdot y_p ) + \FF^p G_0^{n+m})_p \hbar^j$ for $x \in G^n$ and $y
  \in G^m$ and thus for general $x = \sum_j x^{(j)} \hbar^j$ and $y = \sum_k
  y^{(k)} \hbar^k$ with $x^{(j)} = (x_p^{(j)} + \FF^p G_0^n)_p$ and $y^{(k)} =
  (y_p^{(k)} + \FF^p G_0^m)_p$ we obtain the formula
  \begin{align*}
	 x \cdot y = \sum_{k=0}^{+\infty} \bigg( \sum_{l=0}^k p_l \Big(
	 \sum_{j=0}^{k-l} x_{p-m}^{(j)} y_p^{(k-l-j)} \Big) + \FF^p G_0^{n+m}
	 \bigg)_p \hbar^k.
  \end{align*}
\end{rem}
\begin{lem}
  The product on $G_\hbar^{n} \times G_\hbar^m$ is continuous in each entry.
\end{lem}
\begin{proof}
  Let $x_r, x \in G_\hbar^n$ with $x_r \to x$ and $y \in G_\hbar^m$. Write $p_j(x_r) =
  (x_{r,p}^{(j)} + \FF^p G_0^n)_p$, $p_j(x) = (x_p^{(j)} + \FF^p G_0^n)_p$, and
  $p_j(y)= (y_p^{(j)} + \FF^p G_0^m)_p$.  Fix $j, p \in \N_0$.  Take $r_0$ such
  that for all $0 \leq k \leq j$ and all $r \geq r_0$ we have $x_{r,
  p-m}^{(k)} \equiv x_{p-m}^{(k)} \pmod{ \FF^{p-m} G_0^n}$. Such $r_0$ exists
  since $x_r \to x$. For $r \geq r_0$ we have by lemma \ref{lem:multcont} and
  remark \ref{rem:explicitproduct},
  \begin{align*}
	 p_j(x_r \cdot y)-p_j(x\cdot y) = \bigg( \sum_{l=0}^j p_l\Big( \sum_{k=0}^{j-l}
	 (x_{r,p-m}^{(k)} - x_{p-m}^{(k)}) \cdot y_p^{(j-l-k)} \Big) + \FF^p
	 G_0^{n+m}\bigg)_p
	 = 0
  \end{align*}
  Continuity in the other entry follows analogously.
\end{proof}

Now that we have set up the algebra, we define the quantization mapping.  We
already have the map $q_0 : X_0 \to G_0$. Since this map respects the
respective filtrations, we obtain the quantization mapping
\begin{align*}
  q : X \to G \subset G_\hbar 
\end{align*}
as an extension. It is a morphism of graded vector spaces.
\begin{lem}
  The map $q: X \to G$ is an isomorphism of graded vector spaces.
\end{lem}
\begin{proof}
  The inverse of $q_0$ also respects the filtrations and hence extends to a map
  $G \to X$ which is the inverse of $q$.
\end{proof}
To relate the multiplicative structure on $X$ to the one on $G_\hbar$, we set $A =
G_\hbar/(\hbar)$ where $(\hbar)\subset G_\hbar$ is the two sided ideal generated by
$\hbar$. Hence $A \cong G$ as graded vector spaces but not as graded algebras
since $G$ is not closed under multiplication.
\begin{rem}
  \label{rem:cominh}
  We have $[G_\hbar, G_\hbar] \subset (\hbar)$. Hence $(\hbar)$ is the two-sided
  ideal generated by the commutator. Hence the map $\frac{1}{\hbar}[ - , - ] :
  G_\hbar \otimes G_\hbar \to G_\hbar$ is well-defined and turns $G_\hbar$
 	into a graded noncommutative Poisson algebra.  This map descends to $A$.
\end{rem}
\begin{thm}
  \label{thm:comp-brack-comm}
	The graded (commutative) Poisson algebras $X$ and $A$ are isomorphic via the
	map $\phi  = \pi \circ \iota \circ q : X \to G \to G_\hbar \to A$, where $\pi
	: G_\hbar
	\to A=G_\hbar/(\hbar)$ is the canonical projection and $\iota : G \to G_\hbar$ is the
  inclusion. 
\end{thm}
\begin{proof}
  We already know that $\phi$ is an isomorphism of graded vector spaces. By
  remark \ref{rem:cominh}, $A=G_\hbar/[G_\hbar,G_\hbar]$ is graded commutative. Let $x \in X_0^n$
  and $y \in X_0^m$. Write $q(x) = x_- \otimes x_+$ and $q(y) = y_- \otimes
  y_+$. We have $q(xy) = (-1)^{\deg x_+ + \deg y_-} x_- y_- \otimes x_+ y_+$.
  On the other hand $q(x) \cdot q(y) = (x_- \otimes 1) \cdot (1 \otimes x_+)
  \cdot ( y_- \otimes 1) \cdot (1 \otimes y_+)$. The component of $(1 \otimes
  x_+ )\cdot( y_- \otimes 1)$  constant in $\hbar$ equals $(-1)^{\deg x_+ +
  \deg y_-} y_- \otimes x_+$. Hence $q(x) q(y) \equiv q(xy) \pmod{\hbar}$. The
  statement for the completion follows from continuity of $q$ and the product.
  Next we show: 
  For all $x,y \in X$	
  \begin{align}
					\label{eq:commbracket}
    \frac{1}{\hbar}[q(x),q(y)] \equiv q(\{x,y\}) \quad \pmod{ (\hbar)}
  \end{align}
  By density of $X_0 \subset X$ and continuity of all maps involved in the
  statement it suffices to consider $X_0$. We perform an
  induction in the number $n \in \N_0$ of factors of $x$. The anchor $n=0$
  is obvious. We also need the anchor for $n=1$: We prove it by induction on
	the number $m$ of factors in $y$. For $m=0$ nothing is to prove. For
	$m=1$ the statement holds with equality by construction. Now assume the
	statement holds for all $x \in \Sym^1(\NN \oplus \NN^*)$ and all $y \in
	\Sym^m(\NN \oplus \NN^*)$, where $m \geq 1$. Let $x, b \in \Sym^1(\NN \oplus
	\NN^*)$ and $y \in \Sym^m(\NN \oplus \NN^*)$. We have
	\begin{align*}
		\frac{1}{\hbar} [q(x), q(by)] & \equiv \frac{1}{\hbar}[q(x), q(b) q(y)] =
		\frac{1}{\hbar} ( q(b)[q(x), q(y)] (-1)^{\deg x \deg b} + [q(x), q(b)]
		q(y)) \\
		& \equiv
		q(b) q(\{ x,y\}) (-1)^{\deg x \deg b} + q(\{x,b\}) q(y) \\
		& \equiv
		q\Big(b \{ x,y\} (-1)^{\deg x \deg b} + \{x,b\} y \Big) \\
		&\equiv q(\{x,by\})
		\pmod{\hbar}.
	\end{align*}
	Assume now the statement holds for all $x \in \Sym^n( \NN \oplus
	\NN^*)$, where $n \geq 0$. Let $x \in \Sym^n( \NN \oplus \NN^*)$ and $a \in
	\Sym^1( \NN \oplus \NN^*)$. Let $y \in X_0$. We compute similarly
  \begin{align*}
	 \frac{1}{\hbar} [ q(ax), q(y) ] 
	 &\equiv \frac{1}{\hbar} [ q(a)q(x), q(y)] 
	 = \frac{1}{\hbar} q(a) [ q(x),q(y)]  + 
	 \frac{1}{\hbar} [ q(a), q(y)] q(x)
	 (-1)^{\deg x \deg y}\\
	 &\equiv 
  q(a) q(\{ x, y \} ) + 
  q( \{ a, y \})  q(x)
	 (-1)^{\deg x \deg y} \\
	 &\equiv 
  q\Big(a \{ x, y \} + \{ a, y \} x (-1)^{\deg x \deg y}\Big)
  \\
	 &\equiv 
	 q( \{a x, y \})  \quad \pmod{\hbar}
  \end{align*}
\end{proof}
\begin{cor}
\label{cor:rr}
  If $R \in X^1$ solves the classical master equation, then
  $\frac{1}{\hbar} [q(R),q(R)] \equiv 0 \pmod{
					(\hbar) }$. Conversely, if $r = q(R) + \hbar(\cdots)$ solves the
	quantum master equation $[r,r]=0$, then $R$ solves the classical master
	equation.
\end{cor}

\subsection{Solving the Quantum Master Equation}
\label{ssec:qmesolution}

We want to construct a solution $r \in G_\hbar^1$ of the quantum master
equation $[r,r] = 0$. We seek a solution of the form 
\begin{align*}
  r = q(R) + \hbar q(R_1) + \hbar^2 q(R_2) + \cdots
\end{align*}
for some $R_j \in X^1$ where $R \in X^1$ is a given solution of the classical
master equation.

\begin{ass}
\label{ass:coho}
  We assume $H^2(X, d_R) = 0$. 
\end{ass}

\subsubsection{A Differential on the Quantum Algebra}
\label{sssec:diffonqa}

Define $D = \frac{1}{\hbar}[ q(R), -]$. This defines a map $G_\hbar \to
G_\hbar$ by remark \ref{rem:cominh}. It preserves the ideal $(\hbar)$ and
hence descends to a derivation $D_0$ on $A = G_\hbar/(\hbar)$.  We
calculate
\begin{align*}
  D^2(x) = \frac{1}{\hbar^2}[ q(R), [q(R), x]] =
  \frac{1}{\hbar^2} [ [q(R), q(R)], x] - \frac{1}{\hbar^2}
  [q(R) , [q(R), x]]
\end{align*}
Hence by corollary \ref{cor:rr} and remark \ref{rem:cominh},
\begin{align}
  \label{eq:1}
  D^2(x) &= \frac{1}{2\hbar} [ \frac{1}{\hbar}[q(R), q(R)], x]
  \equiv 0 \quad \pmod{ \hbar}
\end{align}
 so $D_0$ is a differential on $A$.
\begin{thm}
  We have $D_0 \circ \phi = \phi \circ d_R$. In particular $\phi : X
  \to A$ is an isomorphism of differential graded commutative
  algebras and $H^\bullet(X, d_R) \cong H^\bullet(A, D_0)$.
\end{thm}
\begin{proof}
  Let $x \in X$. By theorem \ref{thm:comp-brack-comm}
  \begin{align*}
    D_0(\phi(x)) = D_0( \pi( q(x) ) )= \pi( D(q(x))) = \pi(
    \frac{1}{\hbar} [q(R), q(x)] ) = \pi( q( \{R,x\} )) = 
	 \pi(q(d_R(x))) = \phi(d_R(x))
  \end{align*}
\end{proof}
\begin{cor}
  \label{cor:had0}
  Under assumption \ref{ass:coho}, $H^2(A, D_0) = 0$.
\end{cor}

\subsubsection{Construction of a Solution of the Quantum Master Equation}
\label{sssec:qmesolcons}

\begin{thm} Let $n \geq 0$ be an integer.  For $n \geq 1$, assume we have
  constructed $R_1, R_2, \dots, R_n \in X^1$ such that $r_n := q(R) +
  \sum_{l=1}^n \hbar^l q(R_l)$ satisfies $  \frac{1}{\hbar}[r_n, r_n] \equiv 0
  \pmod{ (\hbar^{n+1}) }$. For $n=0$ set $r_n = q(R)$ which also satisfies this
  assumption by corollary \ref{cor:rr}.	We claim, there
  exists $R_{n+1} \in X^1$ such that 
  \begin{align*}
    \frac{1}{\hbar} [r_n + \hbar^{n+1} q(R_{n+1}), r_n +
    \hbar^{n+1} q(R_{n+1}) 
    ] \equiv 0 \quad \pmod{ (\hbar^{n+2}) }
  \end{align*}
\end{thm}
\begin{proof}
  We compute for any $R_{n+1} \in X^1$, using corollary \ref{cor:rr},
  \begin{align*}
   & \frac{1}{\hbar} [r_n + \hbar^{n+1} q(R_{n+1}), r_n +
    \hbar^{n+1} q(R_{n+1}) ] \\
    =&
    \frac{1}{\hbar} [r_n, r_n] + 2 \hbar^{n+1} \frac{1}{\hbar} [r_n,
    q(R_{n+1}) ] + \hbar^{2n+2} \frac{1}{\hbar} [q(R_{n+1}),
    q(R_{n+1})] \\
    \equiv&
    \frac{1}{\hbar} [r_n, r_n] + 2 \hbar^{n+1} \frac{1}{\hbar} [q(R),
    q(R_{n+1}) ] \quad \pmod{ (\hbar^{n+2})}
  \end{align*}
  By the induction assumption, we can write the right hand side as
  $\hbar^{n+1}$ times
  \begin{align*}
    \frac{1}{\hbar^{n+1}} \frac{1}{\hbar}[r_n,r_n] + 2
    D(q(R_{n+1})) \in H
  \end{align*}
  and we want this to be a multiple of $\hbar$. This means we
  need to show that we can pick $R_{n+1} \in X^1$ such that 
  $  \pi( \frac{1}{\hbar^{n+1}} \frac{1}{\hbar}[r_n,r_n] ) + 2 D_0
    (\phi(R_{n+1})) $
  vanishes in $A$. By corollary \ref{cor:had0} and the fact that
  $\phi$ is surjective, it suffices to prove that $    \pi(
  \frac{1}{\hbar^{n+1}} \frac{1}{\hbar}[r_n,r_n] ) $ is $D_0$-closed.
  By the Jacobi identity,
  \begin{align*}
    0 = \frac{1}{\hbar^2} [r_n, [r_n, r_n]] = 
    D(\frac{1}{\hbar} [r_n, r_n]) + \sum_{l=1}^{n} \hbar^l
    \frac{1}{\hbar} [q(R_l), \frac{1}{\hbar}[r_n, r_n]].
  \end{align*}
  By the induction assumption we may divide this equation by $\hbar^{n+1}$ to
  arrive at
  \begin{align*}
	D( \frac{1}{\hbar^{n+1}} \frac{1}{\hbar} [r_n, r_n] ) =
	- \sum_{l=1}^n \hbar^l \frac{1}{\hbar}
    [q(R_l),
    \frac{1}{\hbar^{n+1}} \frac{1}{\hbar}[r_n, r_n]] 
    \equiv 0 \pmod \hbar
  \end{align*}
\end{proof}
Hence we have constructed $r = q(R) + \hbar q(R_1) + \cdots$
with $[r,r]=0$.

\subsection{Uniqueness of the Solution}
\label{ssec:qmeunique}

In this paragraph we consider questions of uniqueness of solutions of the
quantum master equation that arise from quantization of a solution of the
classical master equation. 

\subsubsection{Quantum Gauge Equivalences}
\label{sssec:qgauge}

We define the subspace $K = \{ x \in G_\hbar^0 : p_0(x) \in q(I^{(2)}) \}$. 
\begin{lem}
	\label{lem:Kclosed}
	$K$ is closed under the commutator.
\end{lem}
\begin{proof}
	Let $x,y \in K$. Theorem \ref{thm:comp-brack-comm} allows us to calculate
	\begin{align*}
		\frac{1}{\hbar} [x,y] \equiv \frac{1}{\hbar} [p_0(x),p_0(y)] \equiv q( \{
			q^{-1}(p_0(x)), q^{-1}(p_0(x)) \} ) \pmod{\hbar}.
	\end{align*}
	By lemma \ref{lem:ipclosed}, the claim follows.
\end{proof}
We call elements of $K$
generators of quantum gauge equivalences. Typical elements of $K$ are
quantizations of generators of classical gauge equivalences or any degree zero
multiple of $\hbar$.
In order to exponentiate the Lie algebra $K$ to a group acting on $G_\hbar$ by
isomorphisms of associative algebras, we show that in each degree in $\hbar$
the Lie algebra $L$ acts pro-nilpotent with respect to the filtration $\FF^p
G^n$. 

\begin{lem}
  \label{lem:quantumgauge}
  We define $\ad_a b = \frac{1}{\hbar}[a,b]$ for $a \in K$ and $b \in G_\hbar$.
	The Lie algebra $\ad K \subset \End(G_\hbar)$ acts pro-nilpotently in each
degree of $\hbar$. In particular, it exponentiates to a group of automorphisms of
associative algebras.
\end{lem}
\begin{proof}
	Fix  integers $j \geq 0$ and $k \geq 1$.
	For $i = 1, \dots, k$, take $u_i = u_{i0} + \hbar u_{i1} + \cdots$ where
	$u_{i0} \in q(I^{(2)})$ and $u_{ij} \in G^0$. Let $x \in G^n$. Fix $l = l_1 +
	l_k$ with integers $l_i \geq 0$.  Then 
	\begin{align*}
		&p_j(\ad_{u_1}^{l_1} \cdots \ad_{u_k}^{l_k} x) \\
		=& p_j( \sum_{\substack{j_i: \{1, \dots, l_i\} \to \N_0 \\ i = 1, \dots,
				k}} \hbar^{\sum_{i=1}^k \sum_{s=1}^{l_i} j_i(s)} \ad_{u_{1j_1(1)}} \circ \cdots
		\circ \ad_{u_{1j_1(l_1)}} \circ \cdots \circ \ad_{u_{kj_k(1)}} \circ \cdots
		\circ \ad_{u_{kj_k(l_k)}} (x) ) \\
		=& \sum_{\substack{j_i: \{1, \dots, l_i\} \to \N_0 \\ i = 1, \dots,
				k \\
			n = \sum_{i=1}^k \sum_{s=1}^{l_i} j_i(s) \leq j }}  p_{j-n} (\ad_{u_{1j_1(1)}} \circ \cdots
		\circ \ad_{u_{1j_1(l_1)}} \circ \cdots \circ \ad_{u_{kj_k(1)}} \circ \cdots
		\circ \ad_{u_{kj_k(l_k)}} (x) ) \\
	\end{align*}
	This is a finite sum. We now write out each argument of $p_{j-n}$ as a sum of
	products of the $u_{pq} \in G$. Each such product satisfies the following
	conditions: It contains $(l+1)$ factors. The
	factor $x$ appears once. The number of factors $u_{pq}$ with $q \geq 1$ is
	bounded above by $n$. Hence the number of factors $u_{i0} \in q(I^{(2)})$ is
	bounded below by $(l-n)$. Thus the number of positive factors before normal
	ordering is bounded from below by $2(l-n)$. After normal ordering and
	applying $p_{j-n}$ the number of positive factors that still remain are
	bounded from below by $2(l-n) - (j-n) \geq 2(l-j)$.	
	Hence, $p_j(\ad_{u_1}^{l_1} \cdots \ad_{u_k}^{l_k} x)$ is a finite sum of
	elements in $G^n$, which contain at least $2(l-j)$ factors of positive
	degree. This bound is independent of $x$. 

	Finally fix $p,j \geq 0$ and let $x = \sum_j x_j \hbar^h \in G_\hbar^n$. We have 
	\begin{align*}
	 	p_j( \ad_{u_1}^{l_1} \cdots \ad_{u_k}^{l_k} x) = \sum_{k=0}^j
		p_{j-k} (
		\ad_{u_1}^{l_1} \cdots \ad_{u_k}^{l_k} x_k ).  
	\end{align*}
	Pick $l_0$ such that for all $m = 0, \dots, j$, for all $r \geq 0$ and $l=l_1
	+ \cdots + l_k \geq l_0$ we have $p_m(  \ad_{u_1}^{l_1} \cdots
	\ad_{u_k}^{l_k} x_r ) \in \FF^p G^n$. Then for all $l = l_1 + \cdots + l_k \geq
	l_0$ we have $p_j ( \ad_{u_1}^{l_1} \cdots \ad_{u_k}^{l_k}  x) \in \FF^p
	G^n$.  Hence $\ad K$ acts pro-nilpotently
	with respect to this filtration and thus $\ad K$ exponentiates to a group of
	vector space automorphisms $\{\exp \ad_u : G_\hbar \to G_\hbar, u \in K\}$.
	These maps preserve the multiplicative structure since $\ad_u$ is a derivation
	for the product.
\end{proof}

\subsubsection{Ambiguity for a Given Solution of the Classical Master Equation}
\label{sssec:uniquenessgivenR}

Let $R \in X^1$ be a solution of the classical master equation. 
Throughout this paragraph we assume
\begin{ass}
  \label{ass:uniqueness}
  We have $H^1(X,d_R) = 0$. Thus $H^1(A, D_0) = 0$.
\end{ass}
Let
\begin{align*}
  r & = q(R) + \hbar q(R_1) + \dots \\
  r' & = q(R) + \hbar q(R_1') + \dots
\end{align*}
be two solutions to the quantum master equation, so that $r \equiv r' \pmod{\hbar}$. 
\begin{lem}
  \label{lem:quniquelitestep}
  Let $n \in \N_0$.  Assume that for $l = 1, \dots, n$ we have $R_l = R_l'$.
	Then there exists a generator $c \in (\hbar^{n+1}) \subset K$ of a quantum
	gauge equivalence such that $\exp \ad_c r \equiv r' \pmod{ (\hbar^{n+2}) }$.
\end{lem}
\begin{proof}
	Let $v = q(R_{n+1})-q(R_{n+1}') \in G^1$. Then $0 = [r + r', r-r']$ since
	$r,r'$ solve the quantum master equation. Moreover, $r-r' \equiv \hbar^{n+1}
	v \pmod{\hbar^{n+2}}$. Hence
  \begin{align*}
	 0 = \frac{1}{\hbar} [r+r', v + \hbar \cdots] \equiv \frac{1}{\hbar}[ 2
	 q(R), v] \equiv 2 Dv \pmod{ \hbar}
  \end{align*}
  Thus $D_0 \pi v = 0$.  Hence by assumption \ref{ass:uniqueness}, $\pi v = D_0
	\pi u$ for some $u \in G^0_\hbar$, so $v \equiv D u \pmod{\hbar}$. Since $v
	\in G^1$ is constant in $\hbar$, we may also assume $u \in G^0$. Set $c =
	\hbar^{n+1} u \in K$.  We check that
  \begin{align*}
	 \exp \ad_c r - r' &= r-r' + \frac{1}{\hbar}[c,r] + \sum_{l=2}^{+\infty}
	 \frac{1}{l!} \ad_c^l r \\
	 &\equiv \hbar^{n+1} (v - \frac{1}{\hbar}[r, u]) \equiv \hbar^{n+1} ( v - D
	 u) \equiv 0 \pmod{ (\hbar^{n+2})}.
  \end{align*}
\end{proof}

\begin{thm}
  \label{thm:quniquelite}
  Under assumption \ref{ass:uniqueness}, there is a quantum gauge equivalence
  mapping $r$ to $r'$.  
\end{thm}
\begin{proof}
	By lemma \ref{lem:quniquelitestep}, there exists a sequence of generators
	$c_j \in (\hbar^{j+1}) \subset K$ of quantum gauge equivalences $\exp
	\ad_{c_j}$ which define a sequence $r_{(j)}$ of solutions of the quantum
	master equation via $r_{(0)} = r$ and $r_{(j+1)} = \exp \ad_{c_j} r_{(j)}$,
	so that $r_{(j+1)} \equiv r' \pmod{ \hbar^{j+2}}$. We have
	\begin{align*}
					r_{(j+1)} = \exp \ad_{c_j} \exp \ad_{c_{j-1}} \cdots \exp \ad_{c_0} r
					= \exp \ad_{\gamma_j} r,
	\end{align*}
	for some $\gamma_j \in K$. We are left to show that $\gamma_j \to \gamma \in
	K$ and $\exp \ad_\gamma r = r'$. By the Campbell-Baker-Hausdorff formula, we
	have $\gamma_0 = c_0$ and $\gamma_{j+1} = \gamma_j + c_{j+1} + \cdots$, where
	the terms we have dropped involve sums of nested commutators
	$\frac{1}{\hbar}[-,-]$, each of which containing at least one $c_{j+1} \in
	(\hbar^{j+2})$. Hence, $\gamma_{j+1} \equiv \gamma_{j} \pmod{\hbar^{j+2}}$
	and thus $\lim \gamma_j = \gamma \in K$ exists.

	Finally, fix $k \in \N_0$. We will prove that $p_k(\exp \ad_\gamma r - r') =
	0$. We already know that $p_k( \exp \ad_{\gamma_{k-1}}r - r') = 0$ since $\exp
	\ad_{\gamma_{k-1}} r = r_{(k)}$. Hence it suffices to prove
	that $p_k( \exp \ad_{\gamma_{k-1}}r  - \exp \ad_\gamma r) = 0$. We show by
	induction in $l \in \N_0$ that $\ad_{\gamma_{k-1}}^l r \equiv \ad_\gamma^l r
	\pmod{\hbar^{k+1}}$. The case $l=0$ is trivial. Now suppose the statement
	holds for some $l \in \N_0$. Then 
	\begin{align*} 
					\ad_{\gamma_{k-1}}^{l+1} r - \ad_\gamma^{l+1} r &=
					\ad_{\gamma_{k-1}} ( \ad_{\gamma_{k-1}}^{l} r - \ad_\gamma^{l} r +
					\ad_{\gamma}^l r) - \ad_\gamma^{l+1} r \\
					&\equiv \ad_{\gamma_{k-1}} \ad_\gamma^l r - \ad_\gamma^{l+1} r 
					= \ad_{\gamma_{k-1} - \gamma} \ad_\gamma^l r \equiv 0
					\pmod{\hbar^{k+1}},
	\end{align*}
	since $\gamma \equiv \gamma_{k-1} \pmod{\hbar^{k+1}}$.
\end{proof}

\subsubsection{Ambiguity for Two Classical Solutions Corresponding to the Same
Tate Resolution}
\label{sssec:qsametate}

Let $R, R'$ be two solutions of the classical master equation associated to the
same Tate resolution. Let 
\begin{align*}
  r & = q(R) + \hbar q(R_1) + \cdots \\
  r' & = q(R') + \hbar q(R_1') + \cdots
\end{align*}
be two solutions of the quantum master equation. 
\begin{thm}
  If either of the two solutions $R,R'$ of the classical master equation
	satisfy assumption \ref{ass:uniqueness}, then there is a quantum gauge
	equivalence mapping $r$ to $r'$.
\end{thm}
\begin{proof}
  By theorem \ref{thm:gaugeequivalence}, there exists a classical gauge
  equivalence $g = \exp \ad_u$ mapping $R$ to $R'$. In particular, assumption
  $\ref{ass:uniqueness}$ is satisfied for both solutions. We have $c=q(u) \in
	K$. Set $r'' = \exp \ad_c r$. It is a solution of the quantum master
	equation. We first prove that $r'' \equiv r' \pmod{\hbar}$. We have
	\begin{align*}
					r'' - r' = \exp \ad_c r - r' \equiv \exp \ad_{q(u)} q(R) - q(R')
					\pmod \hbar
	\end{align*}
	Now we prove by induction in $l \in \N_0$ that $\ad_{q(u)}^l q(R) \equiv
	q(\ad_u^l R) \pmod{\hbar}$. For $l=0$ this is obvious. Suppose it holds for
	some $l \in \N_0$. Then, by equation \ref{eq:commbracket},
	\begin{align*}
					\ad^{l+1}_{q(u)} q(R) = \ad_{q(u)} \ad_{q(u)}^l q(R) \equiv
					\ad_{q(u)} q( \ad_u^l R) \equiv q(\ad^{l+1}_u R) \pmod \hbar.
	\end{align*}
	We now have
	\begin{align*}
					\sum_{l=0}^L \frac{1}{l!} \ad_{q(u)}^l q(R) -q(R') \equiv q\bigg(
						\sum_{l=0}^L \frac{1}{l!} \ad_{c_0}^l R  - R'\bigg) \pmod{\hbar}
	\end{align*}
	For $L \to + \infty$ the left hand side converges to $\exp \ad_{q(u)} q(R) -
	q(R')$, and the argument of $q$ on the right hand side converges to zero. By
	continuity of $q$ and $(\hbar)$ being closed, we conclude $r'' \equiv r'
	\pmod{\hbar}$.

	We are now in the situation
	\begin{align*}
					r & = q(R) + \hbar q(R_1) + \cdots \\
					r' & = q(R') + \hbar q(R_1') + \cdots \\
					r'' = \exp \ad_c r & = q(R') + \hbar q(R_1'') + \cdots.
	\end{align*}
	By theorem \ref{thm:quniquelite}, there exists a quantum gauge equivalence
	$\exp \ad_v$ with $\exp \ad_v r'' = r'$. In particular, $r' = \exp \ad_v \exp
	\ad_c r$. 
\end{proof}

\subsubsection{Quantization of Trivial BRST Models}
\label{sssec:trivialquant}

Let $(Y,S)$ be a trivial BRST model, so $Y = \Sym(\NN \oplus \NN^*)$ for some
negatively graded vector space $\NN$ with finite-dimensional homogeneous
components and $S = \sum_j e_j^* \delta(e_j)$ with $\delta (e_j) = e_k$ for
some $k$ depending on $j$. Let $q:Y \to G$ denote the quantization map. Then $s
= q(S)$ solves the quantum master equation.  Moreover, $D_s = \frac{1}{\hbar}[
s, -]$ maps $G_0$ to $G_0$ and $G$ to $G$, as can be seen using the Leibnitz
rule. Hence both $q_0 : Y_0 \to G_0$ and $Y \to G$ are isomorphisms of
differential graded vector spaces. In particular $H^j(G_0, D_s) = 0$ for $j
\neq 0$ and $H^0(G_0, D_s) = \R$ by lemma \ref{lem:trivial-bfv-coho}.

\subsubsection{Quantization of Products With Trivial BRST Models}
\label{sssec:quantproduct}

Let $(X,R)$ be a BRST model and $(Y,S)$ be a trivial BRST model. Consider
quantizations $q : X \to F \subset F_\hbar$ and $q: Y \to G \subset G_\hbar$
with associated solutions of the quantum master equation  $s = q(S)$ and $r =
q(R) + \hbar \cdots$, respectively. Write $F_0 = \Sym(\NN)\otimes \Sym(\NN*)$
for some non-positively graded vector space $\NN$ and $G_0 = \Sym(\UU) \otimes
\Sym(\UU^*)$ for some negatively graded vector space $\UU$. Let $Z = X
\hat{\otimes}Y$ and $q: Z \to H \subset H_\hbar$ be the quantization obtained
from the splitting $H_0 = \Sym( \NN \oplus \UU) \otimes \Sym(\NN^* \oplus
\UU^*)$. 

\begin{lem}
	\label{lem:qquasiiso}
				The natural map $F_\hbar \to H_\hbar$ is a quasi-isomorphism of graded
				associative algebras.
\end{lem}
\begin{proof}
				The natural map is a morphism of graded associative algebras since adding
				new variables from $\UU$ and $\UU^*$ does not change the rules defining
				normal ordering in $F_\hbar$. 

				Consider the isomorphism $\phi_0 : F_0 \otimes G_0 \to H_0$ of graded
				vector spaces. It extends $\hbar$-linearly to an isomorphism 
				\begin{align*}
					\phi : F_0[\hbar] \otimes_{\R[\hbar]} G_0[\hbar] \to H_0[\hbar]
				\end{align*}
				of graded associative algebras. On the left hand side we may take the
				tensor product of algebras, since elements of $F_0$ and $G_0$ commute.

				Using this isomorphism, we construct the $\hbar$-linear maps
				\begin{align*}
					\iota : & F_0[\hbar] \to F_0[\hbar] \otimes_{\R[\hbar]} G_0[\hbar]
					\to H_0[\hbar] \\
					p : & H_0[\hbar] \to F_0[\hbar] \otimes_{\R[\hbar]} G_0[\hbar] \to
					F_0[\hbar]
				\end{align*}
				where the last arrow takes $f\otimes g \in F_0 \otimes G_0$ to $f
				\pi(g)$. Here $\pi : G_0 \to G_0$ is the projection onto $\R$ along
				elements of nonzero form degree.
				
				The quantization map $q : Y_0 \to G_0$ is compatible with the
				decompositions $Y_0 = \R \oplus \bigoplus_{j>0} \Sym^j(\UU \oplus
				\UU^*)$ and $G_0 = \R \oplus \bigoplus_{p+q > 0} \Sym^p(\UU) \otimes
				\Sym^q(\UU)$. Moreover, it intertwines the differentials $d_S$ on $Y_0$
				and $D_s$ on $G_0$. Hence the situation of the proof of lemma
				\ref{lem:quasiiso} is established in the quantum version as well. We
				conclude that $\iota$ and $p$ extend $\hbar$-linearly to the respective
				completions, descend to cohomology, and induce mutual inverses on
				cohomology.
\end{proof}

\subsubsection{Relating Arbitrary BRST Models}
\label{sssec:genunique}

In order to relate the quantizations of two BRST charges defining BRST-models
for the same ideal, we need to quantize general automorphisms of the space $X$
that are the identity on $P=\R[V]$ modulo $I$ (see lemma \ref{lem:philift}).
Since the quantization procedure is not functorial, we do not directly obtain
an isomorphism of differential graded algebras on the quantum level. We were
unable to rigorously define a quantum analog to such an isomorphism.

\newpage
\appendix

\section{Graded Poisson Algebras}
\label{sec:super-poisson}

Let $(P, [-,-]_0)$ be a unital Poisson algebra over $\K = \R$. Let
$\MM$ be a negatively graded vector space with finite-dimensional
homogeneous components. Let $\MM^*$ be the positively graded vector
space with homogeneous components $(M^*)^i = (M^{-i})^*$. Define the
graded algebra $X_0 = P \otimes \Sym(\MM \oplus \MM^*)$.
\begin{lem}
  The bracket $\{-,-\}_0$ on $P$ naturally extends to a skew-symmetric,
  bilinear map $\{-,-\}$ on $X_0$ via the natural pairing of $\MM$ and
  $\MM^*$. This map has degree zero. Moreover, it is a derivation for
  the product on $X_0$ and satisfies the Jacobi identity. Thus, it
  turns $X_0$ into a graded Poisson algebra.
\end{lem}
\begin{proof}
  First we define a bracket on $\Sym(\MM \oplus \MM^*)$. For $x \in
  \MM$ and $\alpha \in \MM^*$ we set
  \begin{align*}
    \{x,x\}_1 & = 0, & \{\alpha, \alpha\}_1 & = 0, &\{x, \alpha\}_1 &= \alpha(x),
    & \{\alpha, x\}_1 &= -(-1)^{\deg \alpha \deg x} \alpha(x)
  \end{align*}
  and extend this definition as a bi-derivation to all of $\Sym(\MM
  \oplus \MM^*)$. It is then a bilinear, skew-symmetric map $\{-,-\}_1 :
  \Sym(\MM \oplus \MM^*) \times \Sym(\MM \oplus \MM^*) \to \Sym(\MM
  \oplus \MM^*)$ of degree zero which is by definition a derivation
  for the product. The expression
  \begin{align*}
    \zeta(a,b,c) := (-1)^{\deg a \deg c} \{a,\{b,c\}\} + \text{ cyclic
      permutations} 
  \end{align*}
  satisfies $\zeta(a_1a_2, b,c) = (-1)^{\deg a_1 \deg c} a_1
  \zeta(a_2,b,c) + (-1)^{\deg a_2 \deg b} \zeta(a_1,b,c)a_2$ and similar
  derivation-like statements for the other entries. Let $\{e_j\}$ be a
  homogeneous basis of $\MM$ and $\{e^*_j\}$ its dual basis. Then the
  bracket of any two of those generators is a scalar, whence $\zeta$
  is zero on generators. By the above derivation-type property,
  $\zeta$ vanishes identically, proving the graded Jacobi identity.

  Now set $\{-,-\} = \{-,-\}_0 + \{-,-\}_1$ on $X_0$.
\end{proof}
The grading on $\MM$ induces a grading on $X_0$. We obtain a filtration:
$\FF^n X_0$ is defined to be the ideal generated by elements of $X_0$
of degree at least $n$. We set $\FF^n X_0^m = \FF^n X_0 \cap
X_0^m$. We also define $I_0 := \FF^1 X_0$ and $I_0^{(n)} := I_0 \cdots
I_0$ to be the $n$-fold product ideal.

\subsection{Compatibility of Filtration and Bracket on $X_0$} 
\label{sec:filtbracket}

We use the derivation properties of the bracket on $X_0$ to derive
compatibility relations between the filtration and the bracket.
\begin{lem}
  \label{lem:optimal}
  For $m, n \in \Z$ and $p,q \in \N_0$ we have $\{\FF^p X_0^n ,\FF^q X_0^m\}
  \subset \FF^{r_{n,m}(p,q)} X_0^{m+n}$ where
  \begin{align}
    \label{eq:optimal}
    r_{n,m}(p,q) = \max\{m+n, \min\{\max\{p, q + n\}, \max\{q, p + m\}\}\}.
  \end{align}
\end{lem}
\begin{proof}
  Let $a, b, u, v \in X_0$ be homogeneous elements with $\deg a + \deg
  u = n$, $\deg b + \deg v = m$, $\deg u = p$, and $\deg v =
  q$. Suppose without loss of generality that $p \geqslant n$ and $q
  \geqslant m$. Then
  \begin{align*}
    \{a u, b v\} = \pm a b \{u, v\} \pm a v \{u, b\} \pm u b \{a, v\} \pm u v
    \{a, b\}.
  \end{align*}
  We now combine different factors to construct elements of high
  degree using the fact that the bracket has degree zero. The first
  summand is in $\FF^{p+q} X_0$. The second summand is in $\FF^{\max\{
    q, p + m\}} X_0$. The third summand is in $\FF^{\max\{ p, n + q\}}
  X_0$. Finally, the last summand is again in $\FF^{p+q} X_0$. So the
  whole sum is in $\FF^r X_0$ where $r = \min\{p + q, \max\{p, q +
  n\}, \max\{q, p + m\}\} = \min\{\max\{p, q + n\}, \max\{q, p +
  m\}\}$.
\end{proof}
\begin{cor} We obtain for $p,q \in \N_0$ and $m,n \in \Z$,
  \label{cor:increaseF}
  \label{cor:x0preservesfp}
  \label{cor:tcor}
  \begin{enumerate}
  \item $\{\FF^p X_0^1, \FF^q X_0^1 \} \subset \FF^{l(p,q)} X_0$, where
    $ l(p,q) =
    \begin{cases}
      \max\{p, q\}, & \text{ if $p \neq q$ } \\
      p + 1, & \text{ if $p=q$ }
    \end{cases}$.
  \item $\{\FF^p X_0, X_0^m\} \subset \FF^p X_0$ provided $m \geqslant 0$.
  \item $\{\FF^p X_0^n, X_0^m\} \cup \{X_0^n, \FF^p X_0^m\} \subset
    \FF^{t_{n,m}(p)} X_0^{n+m}$, where $ t_{n,m}(p) = p - \max\{|n|,|m|\}.  $
  \end{enumerate}
\end{cor}
\begin{lem}
  \label{lem:ipbracket}
  We have $\{X_0^1 \cap \ip, \FF^{m} X_0\} \subset \FF^{m+1} X_0$.
\end{lem}
\begin{proof}
  Let $a, u_1, u_2 \in X_0$ with $\deg(a) = 1-n$,
  $\deg(u_1)+\deg(u_2)=n$, and $\deg(u_1), \deg(u_2) > 0$. Let $b,v
  \in X_0$ with $\deg(v) = m$.  $ \{a u_1 u_2, b v\} = a u_1 \{ u_2, b\} v
  \pm a u_2 \{u_1, b\} v \pm u_1u_2 \{a, b\} v \pm \{a u_1 u_2, v\} b \in
  \FF^{m+1} X_0 $.
\end{proof}
\begin{lem}
  \label{lem:Iclosed}
  The ideal $I_0$ is closed under the bracket.
\end{lem}
\begin{proof}
  Let $a, u, b, v \in X_0$ with $\deg(u)=\deg(v) = 1$. Then $ \{au, bv\}
  = \pm (a b) \{u,v\} \pm (a \{u,b\}) v \pm (\{a, v\} b) u \pm (\{a,b\} u) v
  \in I_0$.
\end{proof}

\subsection{Completion}
\label{sec:compl}

For each $j$, we use the filtration on $X_0^j$ to complete this space
to the space
\begin{align*}
  X^j = \lim_{\leftarrow p} \frac{X_0^j}{\FF^p X_0 \cap X_0^j}.
\end{align*}
The sum and scalar multiplication on $X_0^j$ extend to this space,
turning $X = \bigoplus_j X^j$ into a graded vector space. The product
of two elements $(x_p + \FF^p X_0^j)_p \in X^j$ and $(y_p + \FF^p
X_0^k)_p \in X^k$ is defined to be $(x_p y_p + \FF^p X_0^{j+k})_p \in
X^{j+k}$. This definition does not depend on the choice of
representatives since the product is compatible with the
filtration. Moreover, it defines an element of $X^{j+k}$ since for $p
\leqslant q$ we have $ x_p y_p \equiv x_q y_q \pmod {\FF^p
  X_0^{j+k}}$, since we may shift the representatives of $x$ and
$y$. The multiplication is compatible with the addition turning $X$
into a graded commutative algebra.

Endow $X_0^j / \FF^{p} X_0^j$ with the discrete topology and $\prod_p
X_0^j / \FF^{p} X_0^j $ with the product topology. Equip
$\lim_{\leftarrow} X_0^j / \FF^{p} X_0^j \subset \prod_p X_0^j /
\FF^{p} X_0^j$ with the subspace topology. Finally, equip $X =
\bigoplus_j X^j$ with the product topology. Hence a sequence $\{x_l\}_l
\subset X^j$, with $x_l = ( x_{l,p} + \FF^p X_0^j)_p$, converges to an
element $x = (x_p + \FF^p X_0^j)_p \in X^j$ if and only if for all $p
\in \N_0$ there exists a $l_0$ such that for all $l \geqslant l_0$ we
have $x_{p,l} \equiv x_p \pmod {\FF^p X_0^j}$. A sequence $\{x_l\}_l
\subset X$ converges to an element $x \in X$ if and only if all
homogeneous components converge. Since $X$ is first-countable,
continuity is characterized by the convergence of sequences. We
immediately obtain:
\begin{lem}
  The sum $X \times X \to X$ is continuous.
\end{lem}
For the product, only a weaker statement holds in general:
\begin{lem}
  The product $X \to X$ is continuous in each entry. For each pair
  $(j,k) \in \Z^2$, the product $X^j \times X^k \to X^{j+k}$ is
  continuous.
\end{lem}
\begin{proof}
  Consider a sequence $\{x_i\}_i$ in $X$ converging to $x \in X$ and
  fix $y \in X$. Denote the homogeneous components of $x_{i}$ by
  $x_{i}^j = (x_{i,p}^j + \FF^p X_0^j)_p$ and similarly for $x$ and
  $y$. Fix $l \in \Z$ and $p \in \N_0$. The $l$-th homogeneous
  component of $x_i y$ has a $p$-th component with representative
  $\sum_{j \in C} x_{i,p}^{l-j} y_{p}^{j}$ where the $C \subset \Z$ is
  the finite set for which $y^j \neq 0$. It does not depend on
  $i$. (Such a finite set which is independent of $i$ only exists in
  general when one entry of the product remains fixed.) We have
  \begin{align*}
    \sum_{j \in C} x_{i,p}^{l-j} y_{p}^{j} & = \sum_{j \in C} \Big(
    (x_{i,p}^{l-j} - x_p^{l-j}) y_p^{j} + x_p^{l-j} y_p^{j} \Big)
  \end{align*}
  For each $j \in C$ pick a number $i_{0,j}$ such that for $i_j \gs
  i_{0,j}$ we have $x_{i_j,p}^{l-j} \equiv x_p^{l-j} \pmod { \FF^p
    X_0^{l-j} }$ and let $i_0$ be their maximum. Now, for $i \gs i_0$,
  we have $\sum_{j \in C} x_{i,p}^{l-j} y_{p}^{j} \equiv \sum_{j \in
    C}  x_p^{l-j} y_p^{j} \pmod{ \FF^p X_0^l}$. The second statement
  follows similarly.
\end{proof}
Next, we approximate elements in $X$ by elements in $X_0$.
\begin{lem}
  The map $\iota : X_0^n \longrightarrow X^n$ sending $x \in X_0^n$ to
  $(x + \FF^p X_0^n)_p$ is injective.
\end{lem}
\begin{proof}
  Since $\iota$ is linear, the claim follows from $\bigcap_p \FF^p X_0^n =
  0$. 
\end{proof}
\begin{lem}
  \label{lem:iotalimit}
  For $x = (x_p + \FF^p X_0^n)_p \in X_0^n$ we have $\lim_{m \to \infty}
  \iota(x_m) = x$.
\end{lem}
\begin{proof}
  Fix $p$. For $m \geqslant p$ we have that the $p$-th component of
  $\iota(x_m) - x$ is $x_m - x_p \in \FF^p X_0^n$.
\end{proof}
\begin{cor}
  $X_0$ can be considered a dense subset of $X$.
\end{cor}
Now, we turn to the extension of the bracket to the completion.  Let
$x = (x_p + \FF^p X_0^j)_p \in X^j$ and $y = (y_p + \FF^p X_0^k)_p \in
X^k$. We define $\{x,y\} \in X^{j+k}$ to be the element
\begin{align*}
  ( \{x_{s_{j,k}(p)} , y_{s_{j,k}(p)} \} + \FF^p X_0^{j+k})_p,
\end{align*}
where $s_{j,k}(p) := p + \max\{\vert j\vert,\vert k\vert\}$. This
definition does not depend on the representatives of $x$ and $y$ by
corollary \ref{cor:tcor}, since $t_{j,k}( s_{j,k}(p) ) = s_{j,k}(
t_{j,k}(p)) = p$. Moreover, it defines an element of $X^{j+k}$: For $p
\leqslant q$ we have $\{x_{s_{j,k}(p)},y_{s_{j,k}(p)}\} \equiv
\{x_{s_{j,k}(q)},y_{s_{j,k}(q)}\} \pmod { \FF^{s_{j,k}(p)}}$ since we
may shift the representatives of $x$ and $y$. We extend this bracket
as a bilinear map to $X \times X$.
\begin{lem}
  The extension of the bracket on $X_0$ is a skew-symmetric, bilinear,
  degree zero map on $X$ that satisfies the graded Jacobi identity
  (i.e. the bracket is an odd derivation for itself).
\end{lem}
\begin{proof}
  It is trivial that the extended bracket is a skew-symmetric,
  bilinear degree zero map. These properties follow directly from
  the definitions.

  We prove the graded Jacobi identity. Consider elements
  \begin{align*}
    x& = (x_p + \FF^p X_0^j)_p \in X^j & y &= (y_p + \FF^p X_0^k)_p
    \in X^k & z & =
    (z_p + \FF^p X_0^l)_p \in X^l.
  \end{align*}
  The $p$-th element of $\{y,z\}$ has representative
  $\{y_{s_{k,l}(p)},z_{s_{k,l}(p)}\}$. Hence the $p$-th element of
  $\{x,\{y,z\}\}$ has representative $\{x_{s_{j, k+l}(p)},
  \{y_{s_{k,l}(s_{j, k+l}(p))}, z_{s_{k,l}(s_{j,k+l}(p))} \}$. We now
  want to bound the indices from above by a term which is invariant
  under cyclic permutations of $(j,k,l)$. The function $r_{j,k,l}(p)
  := p + 2 (\vert j \vert + \vert k \vert + \vert l \vert)$ does the
  job. So, $(-1)^{j l} \{x, \{y,z\}\} + (-1)^{k j} \{y,\{z,x\}\} + (-1)^{l k}
  \{z,\{x,y\}\}$ has representative
  \begin{align*}
    (-1)^{j l} \{x_{r_{j,k,l}(p)}, \{y_{r_{j,k,l}(p)},z_{r_{j,k,l}(p)}\}\}
    +  \text{cyclic permutations}
  \end{align*}
  which vanishes by the graded Jacobi identity on $X_0$.
\end{proof}
\begin{lem}
  The bracket on $X$ is a derivation for the product.
\end{lem}
\begin{proof}
  Let
  \begin{align*}
    x& = (x_p + \FF^p X_0^j)_p \in X^j & y &= (y_p + \FF^p X_0^k)_p
    \in X^k & z & =
    (z_p + \FF^p X_0^l)_p \in X^l.
  \end{align*}
  The $p$-th element of $\{xy,z\} - (x \{y,z\} + (-1)^{jk} y \{x,z\})$ has
  representative
  \begin{align*}
    &\{ x_{s_{j+k,l}(p)} y_{s_{j+k,l}(p)}, z_{s_{j+k,l}(p)}\}
    - \bigg(
    x_p \{ y_{s_{k,l}(p)}, z_{s_{k,l}(p)} \}
    + (-1)^{jk}
    y_p \{ x_{s_{j,l}(p)}, z_{s_{j,l}(p)} \}
    \bigg)\\
    \equiv\;&
    \{x_q y_q , z_q\} - (x_q \{y_q,z_q\} + (-1)^{jk} y_q \{x_q,z_q\}) 
    \pmod { \FF^p X_0^{j+k+l}}
  \end{align*}
  where $q := p + \vert m\vert + \vert n\vert + \vert k\vert$ is a
  common upper bound of all indices appearing in the formula. The last
  line vanishes by the derivation property of the bracket on $X_0$.
\end{proof}
\begin{lem}
  \label{lem:bracketcontinuous}
  For each pair $(j,k) \in \Z^2$, the map $\{ - ,-\} : X^j \times X^k
  \to X$ is continuous. The map $\{-,-\} : X \times X \to X$ is
  continuous in each entry.
\end{lem}
\begin{proof}
  Let $x_n = ( x_{n,p} + \FF^p X_0^j )_p \in X^j$ and $y_n = (y_{n,p}
  + \FF^p X_0^k)_p \in X^k$ define two sequences converging to the
  respective elements $x = ( x_{p} + \FF^p X_0^j )_p \in X^j$ and $y =
  (y_{p} + \FF^p X_0^k) \in X^k$. Fix $p$, set $s = s_{j,k}(p)$, and
  pick $n_0$ such that for $n \gs n_0$,
  \begin{align*}
    x_{n,s} &\equiv x_s \pmod {\FF^s X_0^j} & y_{m,s} & \equiv y_s \pmod
    {\FF^s X_0^k}
  \end{align*}
  Let $n \gs n_0$. The $p$-th element of $\{x_n, y_n\} - \{x,y\}$ has the
  representative $\{x_{n, s}, y_{n, s}\} - \{x_s, y_s\} \in
  \FF^{t_{j,k}(s)} \subset \FF^p X_0^{j+k}$ by corollary
  \ref{cor:tcor}.
  
  Now consider a sequence $\{x_i\}_i$ in $X$ converging to $x \in X$
  and fix $y \in X$. Denote the homogeneous components of $x_{i}$ by
  $x_{i}^j = (x_{i,p}^j + \FF^p X_0^j)_p$ and similarly for $x$ and
  $y$. Fix $l \in \Z$ and $p \in \N_0$. Set $C = \{ j \in \Z : y^{j}
  \neq 0 \}$. This is a finite set. The $l$-th homogeneous component
  of $\{x_i,y\}$ has a $p$-th component with representative
  \begin{align*}
    \sum_{j \in C} \{x_{i,s_{l-j,j}(p)}^{l-j}, y^j_{s_{l-j,j}(p)}\}
  \end{align*}
  Set $s = \max\{ s_{l-j,j}(p) : j \in C\}$ and pick $n_0$ such that
  for $n \gs n_0$ and all $j \in C$ we have $x^{l-j}_{n, s} \equiv
  x^{l-j}_s \pmod { \FF^s X_0^{l-j}}$. For such $n$,
  \begin{align*}
    \sum_{j \in C} \{x_{n,s_{l-j,j}(p)}^{l-j}, y^j_{s_{l-j,j}(p)}\}
    \equiv
    \sum_{j \in C} \{x_{n,s}^{l-j}, y^j_{s}\} 
    \equiv
    \sum_{j \in C} \{x_{s}^{l-j}, y^j_{s}\} 
    \pmod { \FF^p X_0^l }
  \end{align*}
\end{proof}

\subsection{The Filtration and the Bracket on the Completion}

The filtration on $X_0$ induces a filtration on the completion with
homogeneous components
\begin{align*}
  \FF^p X^n := \lim_{\leftarrow q} \frac{\FF^p X_0^n}{\FF^{p+q} X_0^n}
  = \{ (x_q + \FF^q X_0^n)_{q \gs p} \in X^n : x_q \in \FF^p X_0^n \}.
\end{align*}
This defines a homogeneous ideal $\FF^p X = \bigoplus_n \FF^p X^n$ in
$X$. We set $I = \FF^1 X$ and 
\begin{align*}
				I^{(n)} = \bigoplus_m\; \lim_{\leftarrow p}
				\frac{X_0^m \cap I_0^{(n)}}{\FF^p X_0^m \cap I_0^{(n)}}.
\end{align*} Those are homogeneous ideals in $X$.
\begin{lem}
  \label{lem:filtrationclosed}
  \label{lem:iphattopclosed}
  For each $j \in \Z$, the sets $\FF^p X^j$ and $\iphat \cap X^j$ are
  closed.
\end{lem}
\begin{proof}
  Consider the first statement. Since $X$ is first-countable, it
  suffices to consider sequences $x_n = ( x_{n,q} + \FF^q X_0^j )_q$
  converging to an $x = ( x_{q} + \FF^q X_0^j )_q$ in $X$ with
  $x_{n,q} \in \FF^p X_0^j$ and show that $x \in \FF^p X^j$. So, fix
  $q \gs p$. Let $n$ be an integer with $x_{n,q} \equiv x_q \pmod
  {\FF^q X_0^j}$.  Then $x_q \equiv x_{n,q} \equiv 0 \pmod { \FF^p
    X_0^j}$.

  Now let $x_n = ( x_{n,p} + \FF^p X_0^j)_p$ be a sequence converging
  to $x = ( x_{p} + \FF^p X_0^j )_p$ in $X$ with $x_{n,p} \in
  \ip$. Fix $p$. For $n$ large enough we may replace $x_p$ by $x_{n,p}
  \in \ip$.
\end{proof}
\begin{lem}
  \label{lem:ipclosed}
  \label{lem:ipnil}
  \label{lem:iqx1}
  Fix $p \in \N_0$.
  \begin{enumerate}
  \item $\{I^{(2)},I^{(2)}\} \subset I^{(2)}$.
  \item $\{\iphat, I^{(p)}\} \subset I^{(p+1)} \subset \FF^{p+1} X$
  \item $\{I^{(p)}, X^1\} \subset I^{(p)}$.
  \end{enumerate}
\end{lem}
\begin{proof}
  The first statement: Consider elements $u = (u_p + \FF^p X_0^{j})_p$
  and $v = (v_p + \FF^p X_0^{k})_p$ of $X$ with $u_p, v_p \in
  \ip$. Then the $p$-th element of $\{u, v\}$ has the representative $\{
  u_{s_{j,k}(p)}, v_{s_{j,k}(p)}\} \in \{\ip, \ip\} \subset \ip $ by the
  Leibnitz rule.
  
  Now the second statement:
  First consider $p=0$. Then by the Leibnitz rule, $\{\ip, X_0\} \subset I_0
  \{I_0, X_0\} \subset I_0$. Now consider $p>0$. By repeated use of the
  Leibnitz rule
  \begin{align*}
    \{\ip, I_0^{(p)}\} \subset \{\ip, I_0\} I_0^{(p-1)} \subset I_0 \{I_0,I_0\}
    I_0^{(p-1)} \subset I_0^{(p+1)}
  \end{align*}
  by lemma \ref{lem:Iclosed}. The statement generalizes to the
  completion, as in the case above.

  The third statement follows analogously by picking representatives.
\end{proof}
\begin{lem}
  \label{lem:filtrationsums}
  Let $l \mapsto q(l)$ define an unbounded non-decreasing function $\N
  \to \N$. Let $x_l = (x_{l,p} + \FF^p X_0^n)_p \in \FF^{q(l)} X^n$ define a
  sequence of elements in $X^n$. Then $\sum_{l=0}^{\infty} x_l$
  converges to an element $x \in X^n$.
\end{lem}
\begin{proof}
  We may suppose $q(l) = l$.  Define $x_p := \sum_{l=0}^{p-1}
  x_{l,p}$. Then $x := (x_p + \FF^p X_0^n)_p$ defines an element of
  $X^n$ since, for $p \leqslant q$, we have
  \begin{align*}
    x_q - x_p = \sum_{l=0}^{q-1} x_{l,q} - \sum_{l=0}^{p-1} x_{l,p} =
    \sum_{l=0}^{p-1} (x_{l,q}-x_{l,p}) + \sum_{l=p}^{q-1} x_{l,q} \in
    \FF^p X_0^n.
  \end{align*}
  We claim that $\sum_{l=0}^k x_l$ converges to $x$ as $k \to
  \infty$. Fix $p$. Let $k \geqslant k_0 := p$. Then the $p$-th
  element of $\sum_{l=0}^k x_l - x$ has representative
  $\sum_{l=0}^k x_{l,p} - x_p = \sum_{l=p}^k x_{l,p} \in \FF^p X_0^n$.
\end{proof}
\begin{lem}
  \label{lem:expansion} Each $H \in X^n$ can be expanded as $H =
  \sum_{p \geq 0} h_p$ with $h_p \in B^p \otimes_P T^{n-p}$.
\end{lem}
\begin{proof}
  Write $H = (x_p + \FF^p X_0^n)_p$ with $x_0 = 0$. Pick a homogeneous basis
	$e_i$ of the underlying graded vector space. Redefine $x_p$
  such that $x_p$ does not contain a monomial in $\FF^p X_0^n$. Set
  $h_p = x_{p+1}-x_p \in \FF^p X_0^n$. It cannot contain a monomial of degree
	$(p+1)$ or higher. Hence $h_p \in B^p
  \otimes_P T^{n-p}$. Then by lemmas \ref{lem:iotalimit} and
  \ref{lem:filtrationsums}, $\sum_p h_p = H$.
\end{proof}
\begin{lem}
  All statements from section \ref{sec:filtbracket} are valid for $X_0$
  replaced by $X$.
\end{lem}
\begin{proof}
  The bracket on $X$ is defined by acting on representatives with the
  bracket of $X_0$ where the statements hold. 
\end{proof}

\subsection{Extension of Maps}

Next, we consider the problem of extending maps on $X_0$ to $X$.
\begin{rem}
  \label{rem:extension-maps}
	A linear map on $X_0$ of a fixed degree preserving the filtration up to a
	fixed shift naturally extends to a linear map on $X$ preserving the
	filtration up to the same shift. This extension is continuous. 
\end{rem}

\subsection{The Associated Graded}

The associated graded $\gr X$ of $X$ is defined as the graded algebra
with homogeneous components $\gr^p X = \FF^p X \big/ \FF^{p+1} X$.  We
have $\gr^0 X = X / I$.
\begin{lem}
  $X/I$ is naturally identified with $P \otimes \Sym(\MM)$.
\end{lem}
\begin{proof}
  Let $x = (x_p + \FF^p X_0^n)_p \in X^n$. Let $u_p \in I_0$ and $z_p \in
  X_0^n$ such that $x_p = u_p + z_p$ and $z_p$ does not contain a
  summand in $I_0$ (or is zero), i.e. $z_p \in \Sym_P(\MM)$. Then $z_p
  - z_1 = x_p - x_1 - (u_p - u_1) \in I_0$. Hence $z_p = z_1$ for all
  $p$. Hence $z := ( z_1 + \FF^p X_0^n)_p \in X^n$ and $x$ define the
  same equivalence class in $X/I$. It is clear that different values
  of $z_1$ yield different equivalence classes.
\end{proof}
\begin{lem}
  \label{lem:associatedgraded}
  There is a natural isomorphism $\gr^\bullet X \cong B^\bullet
  \otimes_P T$ of graded algebras.
\end{lem}
\begin{proof}
  The
  inclusions $B^p \hookrightarrow \FF^p X$ induce a $P$-linear map $B
  \longrightarrow \gr X$. From this, we obtain a map $B \otimes_P T
  \to \gr X$ via $B^p \otimes_P T \ni b \otimes t \mapsto bt \in \gr^p
  X$. The claim
  follows since the monomials in $B^p$ span the free $T$-module $\gr^p
  X$: The image of linearly independent monomials in $B^p$ under the above map
  is obviously $T$-linearly independent. Now for a given $x \in \FF^p X$
  decompose it into homogeneous elements $x^n = (x_{n,q} + \FF^q
  X_0^n)_q \in \FF^p X^n$.  Split $x_{n,q} = b_{n,q} + y_{n,q}$ with
  $y_{n,q} \in \FF^{p+1} X_0^n$ and $b_{n,q} \in \FF^p X_0^n$ does not
  contain a summand in $\FF^{p+1}X_0^n$. Then $b_{n,q} - b_{n,{p+1}}
  \in \FF^{p+1}X_0^n$ for $q > p$, and hence this difference
  vanishes. Set $b^n = (b_{n,p+1} + \FF^q X_0^n)_q$. We have that
  $b^n$ and $x^n$ define the same equivalence class in $\gr^p X^n$ and
  hence $b = \sum_n b^n$ and $x$ define the same equivalence class in
  $\gr^p X$. Each $b^n$ is in the image of $B^p \otimes_P T \to \gr^p
  X$.
\end{proof}

\subsection{Form Degree}

We can filter the algebra $X$ by form degree. For $n \in \Z$ and $j
\in \N_0$, we set $X_0^{n,j} = P \otimes \Sym^j( \MM \oplus \MM^*)
\cap X_0^n$ and define the homogeneous components of $X^{(j)} =
\bigoplus_n X^{n,j}$ to be
\begin{align*}
  X^{n,j} = \lim_{\leftarrow p} \frac{X_0^{n,j}}{ \FF^p X_0^n \cap X_0^{n,j}}
\end{align*}
We have
\begin{lem}
  \label{lem:form-degree-convergence}
  If $x_j \in X^n$ have form degree $j$ then $\sum_j x_j$ converges in
  $X$. 
\end{lem}
\begin{proof}
  Fix $n$. Let $g$ denote the ghost degree and $a$ the anti-ghost
  degree. This means that $g = \deg$ on (homogeneous elements in) $P \otimes
  \Sym(\MM^*)$ and $g = 0$ on $\Sym(\MM)$. Similarly $a = \deg$ on $P \otimes
  \Sym(\MM)$ and $a = 0$ on
  $\Sym(\MM^*)$. Hence $g \geqslant 0$, $a \leqslant 0$ and $a+g = \deg$. We
  decompose a summand $s \in X_0^{n,j}$ of a representative of $x_j$ as $s = a_j \otimes x_{j,1} \dots x_{j,j} \in X_0^{n,j}$ according
  to form degree. Let $l_j$ be the number of factors of positive degree in
  this decomposition. Then $g(x_j) \geqslant l_j$ and $a(x_j) \leqslant
  -(j-l_j)$. Hence
  \begin{align*}
    g(x_j) = a(x_j) + (g(x_j)-a(x_j)) \geqslant a(x_j) + (l_j + (j -
    l_j)) = a(x_j) + j = n + j - g(x_j)
  \end{align*}
  So $g(x_j) \geqslant \frac{1}{2}(n+j)$.
  Set $p(j) = \max\{ k \in \Z : k \leqslant
  \frac{j+n}{2}\}$. We obtain $x_j \in \FF^{p(j)} X^n$. Now apply
  lemma~\ref{lem:filtrationsums}.
\end{proof}
\begin{lem}
  \label{lem:form-degree-decomposition}
  If $x \in X^n$. Then there are $x_j \in X^n$ of form degree $j$ with
  $x = \sum_j x_j$.
\end{lem}
\begin{proof}
  Write $x = \sum_l x^l$ with $x^l \in \FF^l X_0^n$. Expand each $x^l
  = \sum_j x_j^l$ where the sum is finite with $x_j^l$ being of form
  degree $j$. By lemma \ref{lem:filtrationsums}, $x_j =
  \sum_l x_j^l$ converges to an element of $X^n$ of form degree
  $j$. By lemma  
  \ref{lem:form-degree-convergence}, the sum $\sum_j x_j$ converges.
  We have $x = \sum_l \sum_j x_j^l = \sum_j \sum_l x_j^l$, which
  can be verified evaluating both sides modulo $\FF^p$ for general $p$.
\end{proof}
\begin{lem}
  \label{lem:form-degree-sum}
  For $\xi_j \in X^n$ of form degree $j$ with $\sum_j \xi_j = 0$ we
  have $\xi_j = 0$.
\end{lem}
\begin{proof}
	Write $\xi_j = (\xi_{j,p} + \FF^p X_0^n)_p$ with $\xi_{j,p} \in
	X_0^{n,j}$. The $p$-th component of $\sum_j \xi_j$ has representative
	$\sum_{j} \xi_{j,p} \in \FF^p X_0^n$, where this sum is effectively
	finite by the proof of lemma \ref{lem:form-degree-convergence}. We see that
	$\xi_{j,p} \in \FF^p X_0^n$ by expanding in a $P$-basis of $X_0$ consisting of
	monomials in basis elements of the underlying vector space $\MM \oplus
	\MM^*$. 
\end{proof}

\subsection{Symplectic Case}
  
Consider the graded commutative algebra $X_0 = \R[x_i, y_i, e_{j}^{(l)},
{e_{j}^{(l)}}^*]$ where $x_i, y_i$ are of degree zero and
$e_{j}^{(l)}$ and ${e_j^{(l)}}^*$ are of degree $-l$ and $l$,
respectively and there are only finitely many generators of a given degree.
Define a Poisson structure by setting $\{x_i, y_i\} = \delta_{ij}$,
$\{e_j^{(l)}, e_m^{(n)}\} = \delta_{jm}\delta_{ln}$ and setting
all other brackets between generators to zero. We complete the space $X_0$ to
the space $X$ as described above. The partial derivatives 
\begin{align*}
		  \frac{\p}{\p x_i} &= - \{y_i, -\} &
		  \frac{\p}{\p y_i} &= \{x_i, -\} &
		  \frac{\p}{\p e_j^{(l)}} &= (-1)^{l+1}\{ {e_j^{(l)}}^*, -\} &
		  \frac{\p}{\p {e_j^{(l)}}^*} &= \{ e_j^{(l)}, -\}
\end{align*}
are defined via the bracket and hence are all well-defined on $X$. 
\begin{lem}
  \label{lem:generatingfunction}
Let $X_i, Y_i,
  E_j^{(l)}, {E_j^{(l)}}^*$ be elements of $X$ such that the assignments
  \begin{align*}
    (x_i, y_i, e_j^{(l)}, {e_j^{(l)}}^*) &\mapsto
    (x_i, Y_i, e_j^{(l)}, {E_j^{(l)}}^*) \\
    (x_i, y_i, e_j^{(l)}, {e_j^{(l)}}^*) &\mapsto
    (X_i, Y_i, E_j^{(l)}, {E_j^{(l)}}^*) 
  \end{align*}
  both define automorphisms $X \to X$ of graded commutative
  algebras. Then the latter map is a Poisson automorphism if there
  exists an element $S(x_i, Y_i, e_j^{(l)}, {E_j^{(l)}}^*) \in X$ such
  that
  \begin{align*}
        \frac{ \p S}{ \p x_i} & = y_i & \frac{ \p S }{ \p Y_i} & = X_i &
    \frac{\p S}{\p e_j^{(l)}} & = (-1)^l {e_j^{(l)}}^* & \frac{\p S}{\p
      {E_j^{(l)}}^*} & = E_j^{(l)}
  \end{align*}
  Here the partial derivatives with respect to the new variables are defined
  via the chain rule.
\end{lem}
\begin{proof}
  Set $\xi_i^{(0)} = x_i$ and $\xi_j^{(l)} = e^{(l)}_j$ for $l > 0$
  and similarly $\eta^{(0)}_i = y_i$ and $\eta_j^{(l)} = {e^{(l)}_j}^*$. We use
  capital Greek letters $\Xi$ and $H$ for the corresponding transformed
  variables. We express the bracket as
  \begin{align*}
    \{f,g\} 
    =& \sum_{l} (-1)^{l \deg f}
    \sum_j ( (-1)^{l} \frac{ \p f}{ \p \xi_j^{(l)}}\frac{\p g}{\p
      \eta_j^{(l)}} - \frac{\p f}{ \p \eta_j^{(l)}} \frac{\p g}{\p
      \xi_j^{(l)}})
  \end{align*}
  since both sides define derivations which agree on generators. The
  sums converge by lemma \ref{lem:filtrationsums}. We have
  \begin{align*}
    \frac{\p S}{\p \xi_j^{(l)}}= (-1)^l \eta_j^{(l)} \text{ and }
    \frac{\p S}{\p  H_j^{(l)}}   = \Xi_j^{(l)} \text{ so also \quad}
    \frac{ \p \eta_{p}^{(q)}}{ \p \xi_{j}^{(l)}} = 
    \frac{ \p \eta_{j}^{(l)}}{ \p \xi_{p}^{(q)}} (-1)^{q+l+ql}
  \end{align*}
  There are functions $f_{j,l}(\xi, \eta) = \Xi_{j}^{(l)}$ and
  $g_{j,l}(\xi,\eta) = H_j^{(l)}$ realizing the change of
  coordinates so that
  \begin{align*}
    f_{j,l}(\xi, \eta( \xi, H)) & = \frac{\p S}{\p H_j^{(l)}} &
    g_{j,l}(\xi, \eta( \xi, H)) & = H_j^{(l)}
  \end{align*}
  We obtain in the variables $(\xi_j^{(l)}, H_j^{(l)})$ 
  \begin{align*}
    \frac{\p f_{m,n}}{\p \xi_j^{(l)} } + 
    \sum_{p,q} \frac{\p \eta_p^{(q)}}{\p \xi_j^{(l)}}
    \frac{\p f_{m,n}}{\p \eta_{p}^{(q)}}  & = \frac{\p^2 S}{\p
      \xi_j^{(l)} \p H_m^{(n)}}  & & & \\
    \frac{ \p g_{m',n'} }{ \p \xi_{j}^{(l)}}  + \sum_{pq}
    \frac{\p \eta_{p}^{(q)}}{\p \xi_j^{(l)}}
    \frac{\p g_{m',n'} }{\p \eta_p^{(q)}} & = 0 &
    \sum_{pq}\frac{ \p \eta_p^{(q)}}{\p H_{m}^{(n)}}
    \frac{\p g_{m',n'}}{\p \eta_{p}^{(q)}} &= \delta_{mm'} \delta_{nn'}
  \end{align*}
  These expressions make sense in the completion by lemma
  \ref{lem:filtrationsums} since $(j,l,n,m)$ are fixed and the
  $\eta_p^{(q)}$ derivatives are of non-decreasing and unbounded
  degree.  Using those equalities, we calculate in the variables
  $(\xi_j^{(l)}, H_j^{(l)})$ the bracket $ \{f_{m,n}, g_{m',n'}\} $ as
  \begin{align*}
    &\sum_{l} (-1)^{l n} \sum_j ( (-1)^{l}
    \frac{ \p f_{m,n}}{ \p \xi_j^{(l)}}
    \frac{\p g_{m',n'}}{\p \eta_j^{(l)}} 
    -
    \frac{\p f_{m,n}}{ \p \eta_j^{(l)}} 
    \frac{\p g_{m',n'}}{\p \xi_j^{(l)}})\\
    =& 
    \sum_{jl} (-1)^{l(n+1)}
    \frac{\p^2 S}{\p \xi_j^{(l)} \p H_m^{(n)}}
    \frac{\p g_{m',n'}}{\p \eta_j^{(l)}} 
    -
    \sum_{jl} (-1)^{ln} \bigg( 
    \sum_{pq} (-1)^l \frac{ \p \eta_p^{(q)}}{ \p \xi_j^{(l)}} \frac{\p
      f_{m,n} }{ \p \eta_p^{(q)}}     \frac{\p g_{m',n'}}{\p \eta_j^{(l)}} 
    +
    \frac{\p f_{m,n}}{ \p \eta_j^{(l)}} 
    \frac{\p g_{m',n'}}{\p \xi_j^{(l)}}
    \bigg)\\
    =& 
    \sum_{jl} (-1)^{l}
    \frac{\p^2 S}{ \p H_m^{(n)} \p \xi_j^{(l)}}
    \frac{\p g_{m',n'}}{\p \eta_j^{(l)}} 
    -
    \sum_{jl} (-1)^{ln} 
    \sum_{pq}  \bigg( (-1)^l
    \frac{ \p \eta_p^{(q)}}{ \p \xi_j^{(l)}} \frac{\p
      f_{m,n} }{ \p \eta_p^{(q)}}     \frac{\p g_{m',n'}}{\p \eta_j^{(l)}} 
    -
    \frac{\p f_{m,n}}{ \p \eta_j^{(l)}} 
    \frac{\p \eta_{p}^{(q)}}{\p \xi_j^{(l)}}
    \frac{\p g_{m',n'} }{\p \eta_p^{(q)}} 
    \bigg)\\
    =& 
    \sum_{jl} 
    \frac{\p \eta_{j}^{(l)}}{\p H_m^{(n)}}
    \frac{\p g_{m',n'}}{\p \eta_j^{(l)}} 
    -
    \sum_{jlpq} 
    \bigg( (-1)^{ln+l} 
    \frac{ \p \eta_p^{(q)}}{ \p \xi_j^{(l)}} \frac{\p
      f_{m,n} }{ \p \eta_p^{(q)}}     \frac{\p g_{m',n'}}{\p \eta_j^{(l)}} 
    -
    (-1)^{ql + qn+l}
    \frac{\p \eta_{p}^{(q)}}{\p \xi_j^{(l)}}
    \frac{\p f_{m,n}}{ \p \eta_j^{(l)}} 
    \frac{\p g_{m',n'} }{\p \eta_p^{(q)}} 
    \bigg)\\
    =& 
    \delta_{mm'}\delta_{nn'}    -
    \sum_{jlpq} 
    \bigg( (-1)^{ln+l} 
    \frac{ \p \eta_p^{(q)}}{ \p \xi_j^{(l)}} \frac{\p
      f_{m,n} }{ \p \eta_p^{(q)}}     \frac{\p g_{m',n'}}{\p \eta_j^{(l)}} 
    -
    (-1)^{qn + q}
    \frac{\p \eta_{j}^{(l)}}{\p \xi_p^{(q)}}
    \frac{\p f_{m,n}}{ \p \eta_j^{(l)}} 
    \frac{\p g_{m',n'} }{\p \eta_p^{(q)}} 
    \bigg) = \delta_{mm'} \delta_{nn'}\\
  \end{align*}
\end{proof}

\section{Normal Ordering}

Let $\{\xi_i\}_{i \in \N}$ be formal variables of degree $d_i \leq 0$
and $\{\eta_i\}_{i \in \N}$ formal variables of degree $-d_i$, such that there
are only $n_l < + \infty$ many in each degree $l$ and $d_i$ is non-increasing
in $i$. An element of the free algebra $F_0 =
\R \la \xi_i, \eta_i \ra$ is called normal ordered if it belongs to the
subspace spanned by all elements
\begin{align*}
		  \xi_{i_1} \dots \xi_{i_k} \eta_{j_1} \dots \eta_{i_l}
\end{align*}
where $i_n$ and $j_n$ are both nondecreasing in $n$, and $k, l \geq 0$ are
integers.  Let $\hbar$ be a formal variable. Let $(cr) \subset F_0[\hbar]$ be
the two sided ideal generated by
\begin{align*}
		  [\xi_i, \eta_j] - \hbar \delta_{ij}, \qquad 
		  [\xi_i, \xi_j], \qquad 
		  [\eta_i, \eta_j],
\end{align*}
where $[x,y] = xy - (-1)^{\deg x \deg y} yx$.
\begin{thm}
		  \label{thm:pbw}
		  Each class $x + (cr)$ in $F_0[\hbar]/(cr)$ has a unique representative
		  $x' \in F_0[\hbar]$ whose coefficients in $F_0$ are all normal ordered.
		  The representative $x'$ depends $\R[\hbar]$-linearly on $x$.
\end{thm}
\begin{proof}
		  Define the vector space $\gfrak = \Span\{\xi_i, \eta_i, \hbar\}$. The
		  formulae $\{\xi_i, \eta_j\} = \hbar \delta_{ij}$, $\{\xi_i, \xi_j\} =
		  \{\eta_i, \eta_j\} = 0$ define a graded Lie algebra structure on
		  $\gfrak$ with central element $\hbar$. The existence, uniqueness, and
		  $\R$-linearity all follow directly from the
		  Poincar\'e-Birkhoff-Witt theorem applied to  this Lie algebra.
		  Linearity in $\hbar$ now follows directly.
\end{proof}
The assignment $x + (cr) \mapsto x'$ is basis independent in the following way:
Consider the two-sided ideal $(c) \subset F_0[\hbar]$ generated by  
\begin{align*}
		  [\xi_i, \xi_j], \qquad 
		  [\eta_i, \eta_j],
\end{align*}
Let $\alpha : F_0[\hbar] \to F_0[\hbar]$ be an
isomorphism induced by an isomorphism of graded vector spaces, i.e. by mapping
$\xi_i \mapsto \sum_j a_{ij} \xi_j$ and $\eta_i \mapsto
\sum_j b_{ij} \eta_j$ where $a$ is
an invertible block diagonal matrix with finite dimensional blocks and
$b$ is the transpose of the inverse of $a$. In particular the above two sums are finite. 
\begin{lem}
	\label{lem:pbw-base-indep}
	Let $x \in F_0[\hbar]$. Then $\alpha(x') \equiv \alpha(x)' \pmod{ (c) }$. 
\end{lem}
For the proof we need the auxiliary
\begin{lem}
	\label{lem:alphaideal} Let $\alpha$ be as above. We have
	$	\alpha(c)  = (c)$ and $\alpha(cr)  = (cr)$.
\end{lem}
\begin{proof}
	We have $\alpha([\xi_i,\xi_j]) = [\alpha(\xi_i), \alpha(\xi_j)] = \sum_{kl}
	a_{ik} a_{jl} [\xi_k,\xi_l] \in (c)$ since the sums are effectively finite.
	Similarly $\alpha([\eta_i,\eta_j]) \in (c)$. Furthermore 
	\begin{align*}
		\alpha([\xi_i,\eta_j] - \delta_{ij} \hbar) = [\alpha(\xi_i),\alpha(\eta_j)]
		- \delta_{ij} \hbar = \sum_{kl} a_{ik} b_{jl} ([\xi_k,\eta_l] - \delta_{kl}
		\hbar) \in (cr)
	\end{align*}
	since the sums are finite and $b$ is the transpose inverse of $a$. We obtain
	$\alpha(c) \subset (c)$ and $\alpha(cr) \subset (cr)$. In both cases equality
	follows since the inverse of $\alpha$ is also induced by an isomorphism of
	graded vector spaces.
\end{proof}
\begin{proof}[Proof of lemma \ref{lem:pbw-base-indep}]
	We may suppose that $x \in F_0$ is a monomial which is constant in $\hbar$.
	We perform an induction in the number $n$ of factors in $x$. For $n=1$ we
	have $x'=x$, hence $\alpha(x') = \alpha(x)$. Since $\alpha(x)$ is a sum of
	monomials with one factor each, we also have $\alpha(x) = \alpha(x)'$.

	Now assume the statement holds for all monomials with at most $n-1$ factors.
	Let $x$ be a monomial with $n$ factors. Assume first that no $\xi_i$ appears
	in $x$. Then $x \equiv x' \pmod{ (c) }$, so $\alpha(x) \equiv \alpha(x')
	\pmod{ (c) } $ by lemma \ref{lem:alphaideal}. The image $\alpha(x)$ also does
	not contain any $\xi_i$. Hence $\alpha(x) \equiv \alpha(x)' \pmod{ (c) }$.
	Now assume that $x$ does contain some $\xi_j$. If some $\xi_j$ appears in
	front, we can write $x = \xi_j x_0$. We thus have $x' \equiv \xi_j x_0'
	\pmod{(c)}$ and conclude
	\begin{align*}
		\alpha(x') \equiv \alpha(\xi_j) \alpha(x_0')
		\equiv \alpha(\xi_j) \alpha(x_0)' \equiv (\alpha(\xi_i) \alpha(x_0))' =
		(\alpha(x))' \pmod{(c)}
	\end{align*}
	In the first equivalence we have used lemma \ref{lem:alphaideal}, in the
	second the induction assumption, and in the third the fact that
	$\alpha(\xi_i)= \sum_j a_{ij} \xi_j$. Finally, if no $\xi_j$ appears in
	front, we may write $x = \eta_{i_1} \cdots \eta_{i_k} \xi_j y_0 \equiv (-1)^m
	\eta_j^l \xi_j x_0 \pmod{(c)}$ for some $m, l \geq 0$, where the sign comes
	from commuting all $\eta_p$ with $p \neq j$ past $\xi_j$. Hence $x' = (-1)^m
	(\eta_j^l \xi_j x_0)'$. If $l = 0$ we
	directly obtain using the case above
	\begin{align*}
		\alpha(x') = \alpha((-1)^m (\xi_j x_0)') \equiv ( \alpha( (-1)^m \xi_j x_0)
		)' = \alpha(x)' \pmod{ (c)}
	\end{align*}
	since $x \equiv (-1)^m \xi_j x_0 \pmod{ (c)}$ implies $\alpha(x) \equiv \alpha((-1)^m
	\xi_j x_0) \pmod{ (c)}$ by lemma \ref{lem:alphaideal}. If $l > 0$ then use
	$\eta_j^l \xi_j = (-1)^d \xi_j \eta_j^l - f \eta_j^{l-1}
	\hbar$ for $d = \deg \xi_j$  and $f = \sum_{s=1}^l (-1)^{sd}$ to conclude 
	\begin{align*}
		x \equiv (-1)^{m+d} \xi_j \eta_j^l x_0 - (-1)^m f \eta_j^{l-1} x_0 \hbar
		\pmod{ (c)}
	\end{align*}
	We now use the induction assumption and the case where a $\xi_j$ is in front
	to compute
	\begin{align*}
		\alpha(x') &= \alpha\Big( ((-1)^{m+d} \xi_j \eta_j^l x_0 )'\Big) -
		\alpha\Big( ( (-1)^m f \eta_j^{l-1} x_0)'\Big) \hbar \\
		& \equiv
	\Big(\alpha( (-1)^{m+d} \xi_j \eta_j^l x_0 ) \Big)' - \Big(\alpha( (-1)^m f
	\eta_j^{l-1} x_0 ) \Big)' \hbar = ( \alpha( (-1)^m \eta_j^l \xi_j x_0 ))'
	\pmod{(c)}
	\end{align*}
	Since $x \equiv (-1)^m \eta_j^l \xi_j x_0 \pmod{(c)}$, the last term equals
	$(\alpha(x))'$ by lemma \ref{lem:alphaideal}. We are done.

\end{proof}

\bibliography{biblio}{}
\bibliographystyle{plain}

\end{document}